\documentclass[12pt,a4paper]{article}

\pdfoutput=1

\usepackage[utf8]{inputenc}

\usepackage[margin=20mm]{geometry}
\RequirePackage{hyperref}

\usepackage{amsfonts}
\usepackage{amsmath}
\usepackage{amsthm}
\usepackage{amssymb}
\usepackage{bbm}

\usepackage{graphicx}
\usepackage{tikz}
\usetikzlibrary{arrows}
\usepackage{subfigure}
\usepackage{color}
\usepackage{natbib}
\bibpunct{(}{)}{;}{a}{,}{,}

\usepackage{enumitem}
\usepackage[boxed, vlined, linesnumbered]{algorithm2e}
\usepackage{authblk}

\newtheorem{theorem}{Theorem}%[section]
\newtheorem{corollary}{Corollary}%[section]
\newtheorem{proposition}{Proposition}%[section]
\theoremstyle{definition}
\theoremstyle{remark}
\SetCommentSty{textit}
\SetAlFnt{\small}
\SetAlCapFnt{\footnotesize}

\def\convp{\stackrel{\mathrm{P}}{\longrightarrow}}

\DeclareMathOperator{\doop}{do}

\DeclareMathAlphabet{\mathsc}{OT1}{cmr}{m}{sc}

\definecolor{graphcol1}{rgb}{0.89, 0.10, 0.11}
\definecolor{graphcol2}{rgb}{0.22, 0.49, 0.72}
\definecolor{graphcol3}{rgb}{0.30, 0.69, 0.29}
\definecolor{graphcol4}{rgb}{0.60, 0.31, 0.64}

\newlength{\edgelength}
\newcommand{\grline}{\mathbin{\tikz[baseline] \draw[-] (0pt, 0.7ex) -- (\edgelength, 0.7ex);}}

\newcommand{\threegraph}[6]{%
  \begin{tikzpicture}[baseline=(one.base)]
    \node[anchor=base east] (one) at (0, 0) {#1};
    \node[anchor=base west] (two) at (1.2, 0) {#3};
    \node[anchor=base] (three) at (0.6, -0.5) {#5};
    \ifthenelse{\equal{#2}{}}{}{\draw[#2] (one.mid east) -- (two.mid west);}
    \ifthenelse{\equal{#4}{}}{}{\draw[#4] (two) -- (three);}
    \ifthenelse{\equal{#6}{}}{}{\draw[#6] (three) -- (one);}
  \end{tikzpicture}
}

\newlength{\exgredge}
\newenvironment{fournodeex2}{%
\begin{tikzpicture}[baseline=(v1.base)]
  \node (v1) at (0, \exgredge) {$X_1$};
  
  \node (v2) at (-\exgredge, 0) {\textcolor{red}{$X_2=x$}};
  \node (v3) at (\exgredge, 0) {$X_3$};
  \node (v4) at (0, -\exgredge) {$Y$};
}
{\end{tikzpicture}}
\newenvironment{fournodeex3}{%
\begin{tikzpicture}[baseline=(v1.base)]
  \node (v1) at (0, \exgredge) {$X_1$};
  \node (v2) at (-\exgredge, 0) {$X_2$};
  \node (v3) at (\exgredge, 0) {$X_3$};
  \node (v4) at (0, -\exgredge) {$Y$};
}
{\end{tikzpicture}}

\begin{document}

\title{Estimating Causal Effects From Nonparanormal Observational Data}
\author{Seyed Mahdi Mahmoudi,~Ernst Wit}
\affil{Johann Bernoulli Institute for Mathematics and Computer Science,\\
University of Groningen}
%9747 AG Groningen, Netherland}

\date{}
\maketitle

\begin{abstract}
One of the basic aims in science is to unravel the chain of cause and effect of particular systems. Especially for large systems this can be a daunting task. Detailed interventional and randomized data sampling approaches can be used to resolve the causality question, but for many systems such interventions are impossible or too costly to obtain.
Recently, \cite{maathuis2010predicting}, following ideas from \cite{spirtes2000causation}, introduced a framework to estimate causal effects in large scale Gaussian systems. By describing the causal network as a directed acyclic graph it is a possible to estimate a class of Markov equivalent systems that describe the underlying causal interactions consistently, even for non-Gaussian systems. In these systems, causal effects stop being linear and cannot be described any more by a single coefficient. In this paper, we derive the general functional form of causal effect in a large subclass of non-Gaussian distributions, called the nonparanormal. We also derive a convenient approximation, which can be used effectively in estimation. We apply the method to an observational gene  expression dataset.
\end{abstract}

%\begin{keywords}
%\noindent\textbf{Keywords:} 

%  As usual, the \maketitle command creates the title and author/affiliations
%  display 

\maketitle

%  If you are using the referee option, a new page, numbered page 1, will
%  start after the summary and keywords.  The page numbers thus count the
%  number of pages of your manuscript in the preferred submission style.
%  Remember, ``Normally, regular papers exceeding 25 pages and Reader Reaction 
%  papers exceeding 12 pages in (the preferred style) will be returned to 
%  the authors without review. The page limit includes acknowledgements, 
%  references, and appendices, but not tables and figures. The page count does 
%  not include the title page and abstract. A maximum of six (6) tables or 
%  figures combined is often required.''

%  You may now place the substance of your manuscript here.  Please use
%  the \section, \subsection, etc commands as described in the user guide.
%  Please use \label and \ref commands to cross-reference sections, equations,
%  tables, figures, etc.
%
%  Please DO NOT attempt to reformat the style of equation numbering!
%  For that matter, please do not attempt to redefine anything!

\section{Introduction}
\label{s:intro}

Inferring cause-and-effect relationships between variables is of primary importance in many fields of  science. The classical approach for determining such relationships uses randomized experiments where a single or few variables are perturbed. Such intervention experiments, however,
can be very expensive, unethical (e.g. one cannot  force a
randomly selected person to smoke many cigarettes a day) or even infeasible. Hence, it is desirable to
infer causal effects from so-called observational data obtained by observing a system
without subjecting it to interventions. 
%Here, we consider the case with  many variables in the influence diagram and only relatively
%few observational data points. We further assume that the causal influence diagram has
%no directed cycles, i.e., no feedback loops, which may be a severe restriction in certain  applications. 
Although some important concepts
and ideas have been worked out \citep{spirtes1995causal,richardson1996discovery,mooij2011causal},
causal inference allowing for cyclic graphs is still in its infancy.
%, and we consider here
%the  setting where the influence diagram is a directed acyclic graph.

\cite{pearl2009causality} described a do-calculus of causal effects, if the underlying
causal diagram is known. In practice, though,
the influence diagram is often not known and one would like to infer causal effects
from observational data together with the influence diagram.
\cite{spirtes2000causation} introduced methods to estimate causal graphs from observational
data, based on a specified causal influence diagram describing qualitatively the causal
relations among variables. \cite{verma1990equivalence} found that typically groups of causal
graphs give rise to the same distribution of the data, which implies that the generating causal DAG is typically
unidentifiable from the data.
These groups of causal
graphs have characterized Markov equivalence
classes for causal DAGs, which called completed partially directed acyclic graph (CPDAG).
It has presented many  algorithms for constructing and estimated CPDAG in different ways.  There are several constraint-based causal search algorithms such as
 search and score methods \citep{chickering2002learning,chickering2003optimal,verma1990equivalence}, the PC-algorithm
\citep{spirtes2000causation}  and Bayesian methods
\citep[]{heckerman1995learning,spiegelhalter1993bayesian}.\\
The PC-algorithm 
\citep{spirtes2000causation} is one of the main algorithms that try to find equivalence class in two steps: first, by estimating the skeleton using conditional independence tests  and
the characterization of the skeleton; second, orienting as many edges as
possible. \cite{kalisch2007estimating} used PC- algorithm  for Gaussian observations and proved high-dimensional consistency for this algorithm.
\cite{maathuis2009estimating} propose a method that based on 
 estimated causal structure from \cite{kalisch2007estimating}, they could used the interventional distribution in the  Gaussian case to drive causal effect from random varibales. Based on  Gaussian structure, they showed that  one can find  the causal effect by a set of constants. 
\cite{harris2013pc} show that  for wide range of distibution the PC-algorithm has high-dimensional consistency. They use rank-based measures of correlations, such as Spearman’s rank correlation and Kendall’s tau,  in tests of conditional independence. In the terminology of \cite{liu2012high}, this broader class that include marginal Gaussian copula is called   ``nonparanormal distributions.''

In the remainder of the paper, we assume the use of the Rank PC (RPC) algorithm \citep{harris2013pc}, i.e. the PC-algorithm in the nonparanormal context. Based on the estimated CPDAG, it is our aim to derive the concept of a causal effect of $x$ on $y$ as a collection of functions of $x$ and to find a consistent way to estimate them. In Section \ref{sec:causalgraphs}, we introduce the causal graph terminology, a  short description of the intervention calculus and the definition of a causal effect. In Section \ref{sec:nonparanormalCE}, we derive the structure of a causal effect of a nonparanormal causal effect and in Section  \ref{sec:NCE}, we define an convenient estimator. In Section  \ref{sec:simulation-methods}, we evaluate the performance of our method in a simulation study. In Section \ref{sec:real data}, we illustrate the method in a real data example.

\section{Causal effects in causal graphs}
\label{sec:causalgraphs}
%\textcolor{blue}{
In this section we describe the background needed in order to define the notion of a causal effect. We begin by defining causal models through directed graphical models.%}
%
% We then shortly summarize the intervention calculus to connect causal interventions with conditional probabilities, which leads to a natural definition of a causal effect. 
%
%%%%%%%%%%%%%%%%%%%%%%%%%%%%%%%%%%%%%%%%%%%%%%%%%%%%%%%%%%%%%%%%%%%%%%
%\subsection{Directed Graphical Models and d-Separation}
%\label{sec:background}
%%%%%%%%%%%%%%%%%%%%%%%%%%%%%%%%%%%%%%%%%%%%%%%%%%%%%%%%%%%%%%%%%%%%%%%%%%%%%%%%
%
%

% Graph $G$ can be used to define the conditional independences structure among the random variable, that can be
%read off from $G$. 
%
%We make this more precise in the following.  We now introduce graph terminology that we require later, whereby we associate the random variables with the nodes. Most of the definitions can be
%found in \cite{spirtes2000causation} and  \cite{lauritzen1996graphical}.

%These independences are natural if the edges in $E$ encode causal/functional relationships among the random variables $X_i$, and a distribution that
%satisfies them is said to be Markov with respect to $G$. 
% We always assume ${V} =  \{1, 2, \ldots, p\}$ and let the vertices of a graph represent the $p$ random variables $X_1, \ldots, X_p$.
%If both ordered pairs  $(X_i, X_j) $ and $(X_j, X_i) $  belong to $E$,we say that we have an  \textit{undirected} edge between $X_i$ and $X_j$, and denoted by $X_i \grline X_j$.
%\textcolor{blue}{
A \textit{graph} is a pair $ G=({V},{E})$, where $V$ is a finite set of \emph{vertices} ${V} =  \{1, 2, \ldots, p\}$, also called \emph{nodes}, of $G$ and $E$  is a subset of $({V}\times {V})$ of ordered pairs of vertices, called the \emph{edges}  or \emph{links} of $G$. We consider $p$  random variables $X_1, \ldots, X_p$, associated to the vertices. 
If  edge $(X_i, X_j) \in E$ but $(X_j, X_i) \notin E$,  we call the edge \textit{directed} or an \emph{arrow}, denoted by $X_i \rightarrow X_j$. In that case, we also say that $X_i$ is a \emph{parent} of $X_j$, and that $X_j$ is a \emph{child} of $X_i$. The set of parents of a vertex $X_j$ is denoted by pa($X_j$).  We use the short-hand notation 
%$X_i \grarright X_j \in G$ to denote $(X_i, X_j) \in E$, but $(X_j, X_i) \notin E$, and 
$X_i \grline X_j$ to denote $(X_i, X_j) \in E$ and $(X_j, X_i) \in E$. 
A graph containing only directed edges ($\rightarrow$) is \emph{directed},
one containing only undirected edges ($\grline $) is \emph{undirected}. A  directed graph is called a
\emph{directed acyclic graph }(DAG) if it does not contain
directed cycles. A common tool for describing equivalence classes of DAGs are completed partially directed
acyclic graphs (CPDAG).%}
\cite{pearl2009causality} defined causality through intervention, whereby variables are externally manipulated to take certain values. This intervention changes the underlying distribution $P$ and can be expressed by adapting the direct effect diagram. The new distribution is called the \emph{intervention distribution} and we say that the variables, whose structural equations we have replaced have been ``intervened on.''
The intervention distribution of $Y$ when doing an intervention and setting the variable $X_i$ to a value $x_i^{\prime}$ is denoted by $P(Y|\doop(X_i = \it{x_i^{\prime}}))$.
%\begin{eqnarray}
%P(Y|\doop(X_i = \it{x_i^{\prime}})).
%\label{eqn:Intervention1}
%\end{eqnarray}
The intervention on variable $X_i$ is characterized by a \emph{truncated factorization}, in which an intervention
DAG $G'$, arising from the non-intervention DAG $G$ can be defined by deleting all edges
which point into the node $X_i$. For an example, In below graphs, a DAG $G$ and  its corresponding intervention graphs ($G'$) are shown.%}

%  \begin{center}
%        \begin{fournodeex3}
%      \draw[->] (v1) -- (v2);
%      \draw[->] (v1) -- (v3);
%      \draw[->] (v2) -- (v4);
%      \draw[->] (v3) -- (v4);
%    \end{fournodeex3}
%    \hspace{1cm}
% % } \qquad \qquad
% % \subfigure[$G'$]{%
%    \begin{fournodeex2}
%      \draw[->] (v1) -- (v3);
%      \draw[->] (v2) -- (v4);
%      \draw[->] (v3) -- (v4);
%    \end{fournodeex2}
% % } \qquad \qquad
%  
% % \caption{(a) A DAG $G$ and (b) its corresponding intervention graphs $G^{\prime}$. The intervention is $\doop(X_2 = x)$, described by the red label in the graph. The parental set of $i = 2$ is
%%$pa_2 = \{1\}$ which appears in (5) for computing the causal effect $\beta_2$ of $Y$ on $X_2$. 
%  %}
%  %\label{fig:ex-intervention-dags}
%%\end{figure}
%\end{center}
%%\textcolor{blue}{

\begin{figure}[tbh]
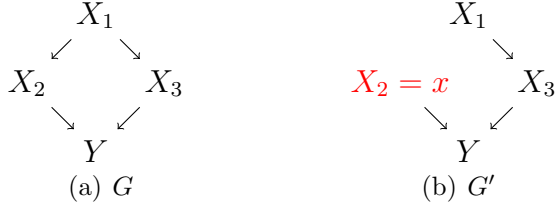

  \centering
  \subfigure[$G$]{%
    \begin{fournodeex3}
      \draw[->] (v1) -- (v2);
      \draw[->] (v1) -- (v3);
      \draw[->] (v2) -- (v4);
      \draw[->] (v3) -- (v4);
    \end{fournodeex3}
  } \qquad \qquad
  \subfigure[$G'$]{%
    \begin{fournodeex2}
      \draw[->] (v1) -- (v3);
      \draw[->] (v2) -- (v4);
      \draw[->] (v3) -- (v4);
    \end{fournodeex2}
  } \qquad \qquad
  
  \caption{(a) A DAG $G$ and (b) its corresponding intervention graphs $G^{\prime}$. The intervention is $\doop(X_2 = x)$, described by the red label in the graph. The parental set of $i = 2$ is
$\mbox{pa}{(2)} = \{1\}$ which appears in (\ref{eqn:causal effect}) for computing the causal effect $\beta_2$ of $Y$ on $X_2$. 
  }
  \label{fig:ex-intervention-dags}
\end{figure}

% The intervention is $\doop(X_2 = x)$, described by the red label in the graph. The parental set of $i = 2$ is
%$\mbox{pa}{(2)} = \{1\}$, which appears in (\ref{eqn:causal effect}) for computing the causal effect $\beta_2$ of $Y$ on $X_2$. 
%Using the Markov property of $P$ with respect to the DAG $G$ and  applying it to the intervention
%graph $G'$, whereby we require that  Markov property is inherited for $P(Y|\doop(X_i = .))$ with respect to
%$G'$, we obtain the unique distribution
%\begin{equation}
%P_Y(y|\doop(X_i = \it{x}_i^{\prime}))
%=\left\{\begin{array}{ll}
%P(y)&  \mbox{if } Y \in x_{\mbox{\scriptsize pa}(i)}\\
%& \\
% \int P(y|x_i^{\prime},x_{\mbox{\scriptsize pa}(i)})\ P(x_{\mbox{\scriptsize pa}(i)})\ d (x_{\mbox{\scriptsize pa}(i)})&  \mbox{if } Y\notin x_{\mbox{\scriptsize pa}(i)}\hspace{3.65cm}\\
%\end{array}\right.
%\label{eqn:Intervention}
%\end{equation}
%This relationship shows that the intervention distribution can be inferred
%from the observational distribution $P$ and the corresponding DAG $G$. All that is
%required is the Markov condition of $P$ with respect to $G$.
%
%
%
%
%
%%%%%%%%%%%%%%%%%%%%%%%%%%%%%%%%%%%%%%%%%
%\subsection{Causal Effect}
%\label{sec:causal effect}
%%%%%%%%%%%%%%%%%%%%%%%%%%%%%%%%%%%%%%%%
The \emph{causal effect} of $X_i$ on $Y$ at a point $x_i^{\prime}$ by the way $Y$ is expected to change as a result from a small interventional  change of $X_i$ at $x_i$,%}
\begin{equation}
\mbox{CE}(Y| X_i = x_i) = {\frac{\partial}{\partial x}} E[Y|\doop(X_i = \it{x})]|_{x = x_i^{\prime}}
\label{eqn:causal effect}
\end{equation}
where we have that 
\begin{equation}
E(Y|{\doop(X_i = \it{x})})
=
\int E(Y|x,x_{\mbox{\scriptsize pa}(i)})\ P(x_{\mbox{\scriptsize pa}(i)})\ d(x_{\mbox{\scriptsize pa}(i)}) \hspace{10mm}  \mbox{if}\ Y\notin {x_{\mbox{\scriptsize pa}(i)}}
\label{eqn:Intervention effect}
\end{equation}
If $( X_1,...,X_{p-1},Y)$ has a multivariate Gaussian distribution, it is very simple to compute the causal effects as defined in (\ref{eqn:causal effect}).  
%
%Gaussianity implies that 
%\begin{equation}
%E\left[Y|X_i=x_i; X_{\mbox{\scriptsize pa}(i)} = x_{\mbox{\scriptsize pa}(i)}\right]=\beta_{0}+\beta_{i} x_i+ \beta_{\mbox{\scriptsize pa}(i)}^T x_{\mbox{\scriptsize pa}(i)}
%\label{eqn: linear causality for Gaussian1}
%\end{equation}
%is linear in $x_i$ and $x_{\mbox{\scriptsize pa}(i)}$. For some $\beta_{0},\beta_{i} \in \mathbb{R} $ and $\beta_{\mbox{\scriptsize pa}(i)}^T \in \mathbb{R}^{\vert {\mbox{\scriptsize pa}(i)} \vert} $, where ${|{\mbox{pa}(i)} |}$ is
%the cardinality of the set ${\mbox{pa}(i)}$. 
Therefore, we have 
\begin{equation}
E(Y|{\doop{(X_i = \it{x}_i)}})
=\beta_i x_i +\int \beta_{{\mbox{\scriptsize pa}(i)}}^T x_{\mbox{\scriptsize pa}(i)} P(x_{\mbox{\scriptsize pa}(i)})\ d(x_{\mbox{\scriptsize pa}(i)}) 
\label{eqn:Intervention effect for normal}
\end{equation}
 is linear in $x_i$, if $  Y\notin {x_{\mbox{\scriptsize pa}(i)}}$
and then the intervention effect, or causal effect, becomes 
\begin{equation}
\mbox{CE}(Y| X_i = x)={\frac{\partial}{\partial x}} E[Y|\doop(X_i = \it{x})]|_{x = x_i^{\prime}} = \beta_i
\label{eqn:causal effect for Gauss}
\end{equation}
A simple way to obtain the parameter $\beta_i$ is given by Pearl’s backdoor criterion
\citep{pearl2009causality}.  From (\ref{eqn:causal effect for Gauss}), it follows that the causal effect
of $X_i$ on $Y$ with  $  Y\notin {x_{\mbox{\scriptsize pa}(i)}}$ is given by the regression
coefficient of $X_i$ in the regression of $Y$ on $X_i$ and ${\mbox{pa}(i)}$. 
Note that if $  Y\in {x_{\mbox{\scriptsize pa}(i)}}$, the causal effect from $X_i$ to $Y$ is, obviously, zero. Our aim is to generalize this to a wider class of distributions. 

%%%%%%%%%%%%%%%%%%%%%%%%%%%%%%%%%%%%%%%%%%%%%%%%%%%%%%%%%%%%%%%%%%%%%
\section{Causal effect for nonparanormal graphical models}
\label{sec:nonparanormalCE}
%%%%%%%%%%%%%%%%%%%%%%%%%%%%%%%%%%%%%%%%%%%%%%%%%%%%%%%%%%%%%%%%%%%%%
\cite{kalisch2007estimating} use the PC-algorithm in a Gaussian setting %\citep{spirtes2000causation}
 for estimating the causal skeleton and, subsequently, the equivalence class of high-dimensional causal graph. The algorithm is based on a clever hierarchical scheme for  testing conditional independences among pairs of variables $X _j , X_k$ (for all $ j \not= k$) in the DAG. 
In Gaussian models, tests of conditional independence can be 
based on Pearson correlations, and high-dimensional consistency results have been obtained for the
PC-algorithm in this setting.

Building on this work, \cite{maathuis2009estimating} are  interested in estimating the causal effect of a covariate
$X_i $ on a response $Y$ in a Gaussian causal graph. After obtaining the equivalence class of causal DAG,  they apply for each DAG $G_j$ in this class the  intervention calculus to obtain
the causal effect $\beta_ {ij}$ of $X_i$ on $Y$, which can easily be shown to be the regression coefficient in 
\begin{equation}
E\left[Y|X_i=x_i; X_{\mbox{\scriptsize pa}(i)} = x_{\mbox{\scriptsize pa}(i)}\right]=\beta_{0j}+\beta_{ij} x_i+ \beta_{\mbox{\scriptsize pa}(i),j}^T x_{\mbox{\scriptsize pa}(i)}
\label{eqn: linear causality for Gaussian}
\end{equation}
where ${\mbox{pa}(i)}$ is the parental index set of $X_i$ in graph $G_j$, 
and then summarize this information for
$i = 1, \ldots, p $ and $j = 1,\ldots, m$ in a $p \times m$ matrix $\bf{\Theta}$.
% where each row corresponds
%to a covariate and each column corresponds to a DAG in the equivalence
%class. Since the ordering of the DAGs in the equivalence class is arbitrary,
%the columns of this matrix can be permuted in any order. In other words, they define multi-sets $\Theta_i = \{\beta_{ij} \}_{j\in\{1,...,m\}}$ containing the estimated possible causal
%effects of  $ X_i$ on $Y$. We note that $\bf{\Theta}$ contains slightly more information than the multi-sets $\Theta_i$ ($ i = 1, . . . , p$), since
%the columns of $\bf{\Theta}$ tell us which possible causal effects originate from the
%same DAG, while this information is lost in the multi-set notation.
%
%\cite{harris2013pc} proved  high-dimensional consistency 
%properties for a broader class of nonparanormal models when using rank-based measures of correlation. They showed that the \textit{ Rank PC-algorithm} (RPC) works as well as the Pearson PC-algorithm for normal data and considerably better for non-Gaussian data. If one assumes to know all conditional independencies exactly, then the RPC-algorithm yields the ``true'' CPDAG, i.e. the equivalent class of DAGs that contains the true causal DAG. 

In this section, we prove how based on this CPDAG we can derive the analogous multi-set of causal effects for Gaussian copula, also called \emph{nonparanormal}, distributed data. In practice, the conditional independences have to be inferred from the data as well and we show how using our main result in combination with the RPC-algorithm we are able to define an convenient estimator for the causal effect for such data, which stops being linear and needs to be estimated functionally.

%%%%%%%%%%%%%%%%%%%%%%%%%%%%%%%%%%%%%%%%%%%%%%%%%%%%%%%%%%%%%%%%%%%%%%%%%%%%%%%
%\subsection{Causal Effect for Nonparanormal Graphical Models}
%\label{sec:Causal Effect for Nonparanormal Graphical Models}
%%%%%%%%%%%%%%%%%%%%%%%%%%%%%%%%%%%%%%%%%%%%%%%%%%%%%%%%%%%%%%%%%%%%%%%%%%%%%%%
\subsection{General expression of nonparanormal causal effect}

\cite{liu2012high}  define the nonparanormal distribution. Let $f=(f_i)_{i\in \mathbf{V}}$ be a set of monotone, univariate functions and let $ \Sigma\in \mathbb{R}^{\mathbf{V}\times \mathbf{V}} $ be a positive definite covariance matrix. We say a $p$-dimensional random variable $X=(X_1,...,X_{p})^\mathbf{T}$ has a nonparanormal distribution, $$X\thicksim \mbox{NPN}(\mu,\Sigma,f),$$ if $f^{-1}(X)=(f_1^{-1}(X),\ldots,f_p^{-1}(X))\thicksim N(\mu,\Sigma) $. If $X\thicksim \mbox{NPN}(\mu,\Sigma,f)$, then the univariate marginal distribution for a coordinate, say $X_i$, can have any distribution $F_i$, as we can take $f_i=F_{i}^{-1}\circ \Phi_{\mu_i,\sigma_i^2}$, where $\Phi_{\mu_i,\sigma_i^2}$ is the normal distribution function with mean $\mu_i$ and variance $\sigma_i^2=\Sigma_{ii}$. Note that $f_i$ need not be continuous.
In this paper, we deal with monotone and differentiable $f$.  \cite{liu2012high} show that in that case the nonparanormal distribution $\mbox{NPN}(\mu,\Sigma,f)$ is a Gaussian copula.

In the remainder of the paper, we consider that $(X_1, \ldots,X_{p-1},Y)\thicksim \mbox{NPN}(0,\Sigma,f)$, where $\Sigma$ is a correlation matrix. We will refer to the latent standard normally distributed variables as $Z_i=f_{i}^{-1}(X_i) = \Phi^{-1}\circ F_i(X_i)$ and $Z=f_{y}^{-1}(Y) = \Phi^{-1}\circ F_y(Y)$.  We are interested in the causal effect of $X_i$ on $ Y$ for $i\in (1,\ldots,p-1)$. 
We know from Section (\ref{sec:causalgraphs}) that for Gaussian data it is very simple to compute the causal effect, since Gaussianity implies that $ E(Y|X_i=x_i; X_{-i}=x_{-i})$ is linear in $x_i$. Unfortunately, this is no longer true for non-Gaussian random variables. In Theorem \ref{thm:npnCE} we derive the explicit functional form for the causal effect in the entire class of nonparanormal distributions. 

\begin{theorem}
\label{thm:npnCE}
Let $(X_1, \ldots,X_{p-1},Y)\thicksim \mbox{NPN}(0,\Sigma,f)$ and $f_i$ $(i=1,\ldots,p-1)$ is differentiable and  $f_y$ is infinitely differentiable, then the causal effect of $X_i$ on $Y$ in causal graph $G$ is given by
\begin{eqnarray}
%{\frac{\partial}{\partial y_i}}E[Y_j|{\doop(Y_i = y_i)}]
\mbox{CE}(Y| X_i = x_i) &=& \sum_{k=1}^{\infty}\sum_{r=0}^{\lfloor \frac{k-1}{2}\rfloor}\sum_{s=1}^{k-2r} f_y^{(k)}(z_{0})\frac{1}{k!} {k-2r\choose s} {k\choose 2r}s\beta_i(-z_0 +\beta_i z_i)^{s-1} \label{eq:ce} \\\nonumber
 &\times & E[(\beta_{{\mbox{\scriptsize pa}(i)}}^{T} Z_{\mbox{\scriptsize pa}(i)})^{k-2r-s}]
  (2r-1)\ldots 3.1\times[(1-\rho ^2)]^r(f_i^{-1})^{\prime}(x_i) 
\end{eqnarray}
for every $z_0\in \mathbb{R}$, where $f_y^{(k)}$ is the $k$th derivative of $f_y$,$z_i = f_i^{-1}(x_i)$, $ Z_{{\mbox{\scriptsize pa}(i)}}= f_{\mbox{\scriptsize pa}(i)}^{-1}( X_{\mbox{\scriptsize pa}(i)})$, $(\beta_i,\beta_{{\mbox{\scriptsize pa}(i)}}) = \Sigma_{p,(i,\mbox{\scriptsize pa}(i))}\Sigma_{(i,\mbox{\scriptsize pa}(i)),(i,\mbox{\scriptsize pa}(i))}^{-1}$ and $\rho=\Sigma_{p,(i,\mbox{\scriptsize pa}(i))}\Sigma_{(i,\mbox{\scriptsize pa}(i)),(i,\mbox{\scriptsize pa}(i))}^{-1}\Sigma_{(i,\mbox{\scriptsize pa}(i)),p}$. 
\end{theorem}
\begin{proof}
%If $(X_1, \ldots,X_{p-1},Y)\thicksim \mbox{NPN}(0,\Sigma,f)$, then the univariate marginal distribution for coordinate , $Y_\nu$, may have any continuous cumulative distribution function $F$ and we  take $f_i=F^{-1}_{i}\circ \Phi$, where $\Phi$ is the standard normal distribution function.  
  We follow three steps for proving this theorem. First, we find a closed form expression for $E\left[Y|X_i=x_i; X_{\mbox{\scriptsize pa}(i)} = x_{\mbox{\scriptsize pa}(i)}\right]$. After that we connect this to the do-operator as is done in (\ref{eqn:Intervention effect}). Finally, taking the derivative in the way that the causal effect is defined in (\ref{eqn:causal effect}) will complete the proof. From the differentiability of $f_i$ follows that  the marginal distributions $F_i$ are one-to-one, where $f_i^{-1}(x_i)=z_i$ and $Z_i=f_{i}^{-1}(X_i) = \Phi^{-1}\circ F_i(X_i)$ and $Z=f_{y}^{-1}(Y) = \Phi^{-1}\circ F_y(Y)$. Using the Taylor expansion,
 \begin{eqnarray}\nonumber
E\left[Y|X_i=\it{x_i}; X_{\mbox{\scriptsize pa}(i)} = x_{\mbox{\scriptsize pa}(i)}\right]&=&E(F_y^{-1}(\Phi(Z))|X_i=x_i; X_{\mbox{\scriptsize pa}(i)} = x_{\mbox{\scriptsize pa}(i)})\\\nonumber
&= &E(F_y^{-1}(\Phi(Z))|Z_i=z_i; Z_{{\mbox{\scriptsize pa}(i)}} = z_{{\mbox{\scriptsize pa}(i)}})\\\nonumber
&= &E(f_y(Z)|Z_i=z_i; Z_{{\mbox{\scriptsize pa}(i)}} = z_{{\mbox{\scriptsize pa}(i)}})\\\nonumber
&= & E(\sum_{k=1}^{\infty}f_y^{(k)}(z_{0})\frac{(Z-z_{0})^k}{k!}|Z_i=z_i; Z_{{\mbox{\scriptsize pa}(i)}} = z_{{\mbox{\scriptsize pa}(i)}})\\
&=&\sum_{k=1}^{\infty}f_y^{(k)}(z_{0})\frac{1}{k!}E({Z^*}^k|Z_i=z_i; Z_{{\mbox{\scriptsize pa}(i)}} = z_{{\mbox{\scriptsize pa}(i)}})
\label{eqn:Taylor expansion}
\end{eqnarray}
where $Z^*=Z-z_{0}$ for any $z_0 \in \mathbb{R}$. From the conditional normal distribution, we know that
$$Z^*|Z_i=z_i; Z_{{\mbox{\scriptsize pa}(i)}} = z_{{\mbox{\scriptsize pa}(i)}} \sim N(-z_0 +  (\beta_i,\beta_{\mbox{\scriptsize pa}(i)}) (z_i, z_{\mbox{\scriptsize pa}(i)})^T,(1-\rho ^2)).$$  
where $(\beta_i,\beta_{{\mbox{\scriptsize pa}(i)}}) = \Sigma_{p,(i,\mbox{\scriptsize pa}(i))}\Sigma_{(i,\mbox{\scriptsize pa}(i)),(i,\mbox{\scriptsize pa}(i))}^{-1}$ and $\rho=\Sigma_{p,(i,\mbox{\scriptsize pa}(i))}\Sigma_{(i,\mbox{\scriptsize pa}(i)),(i,\mbox{\scriptsize pa}(i))}^{-1}\Sigma_{(i,\mbox{\scriptsize pa}(i)),p}$. 
Following \cite{lehmann1998theory} page 132, we get for $k\in\mathbb{N}$
 %&=&f_j(z_{0j})+ f_j^{\prime}(z_{0j})[\beta_0+\beta_i z_i+\beta_{pa_i}^T pa_i -z_{0j}]\hspace{19mm}(\mbox{From (\ref{eqn: linear causality for Gaussian})})\\
%&=&C_1+ f_j^{\prime}(z_{0j})\beta_i f_i^{-1}(y_i)\hspace{55mm}(f_i^{-1}(y_i)=z_i)
 \begin{eqnarray}\nonumber
E({Z^*}^k|Z_i=z_i; Z_{{\mbox{\scriptsize pa}(i)}} = z_{{\mbox{\scriptsize pa}(i)}})&=& \sum_{r=0}^{\lfloor \frac{k}{2}\rfloor} {k\choose 2r}(-z_0 +\beta_i z_i+\beta_{{{\mbox{\scriptsize pa}(i)}}}^{T}z_{{\mbox{\scriptsize pa}(i)}})^{k-2r}\\  
 &\times &  (2r-1)\ldots 3.1\times[(1-\rho ^2)]^r 
 \label{eqn:Lehmann and Casella}
\end{eqnarray}
 With replacement (\ref{eqn:Lehmann and Casella}) in (\ref{eqn:Taylor expansion}) we have
 \begin{eqnarray}\nonumber
E(Y|X_i=\it{x_i}; X_{\mbox{\scriptsize pa}(i)} = x_{\mbox{\scriptsize pa}(i)})&=&  \sum_{k=1}^{\infty}\sum_{r=0}^{\lfloor \frac{k}{2}\rfloor} f_y^{(k)}(z_{0})\frac{1}{k!} {k\choose 2r}(-z_0 +\beta_i z_i+\beta_{{\mbox{\scriptsize pa}(i)}}^{T}z_{{\mbox{\scriptsize pa}(i)}})^{k-2r}\hspace{-2cm}\\ 
 &\times &  (2r-1)\ldots 3.1\times[(1-\rho ^2)]^r
 \label{eqn:Conditional Expectation for nonparanormal} 
\end{eqnarray}
Now we use (\ref{eqn:Conditional Expectation for nonparanormal}) for finding the intervention effect for nonparanormal variable. That is, 
\begin{eqnarray}
E(Y|{\doop(X_i = \it{x_i})})
&=&\int E(Y|X_i=\it{x_i}; X_{\mbox{\scriptsize pa}(i)} = x_{\mbox{\scriptsize pa}(i)})P(x_{\mbox{\scriptsize pa}(i)})\ d(x_{\mbox{\scriptsize pa}(i)}) \hspace{8mm} \mbox{if}\ Y\notin x_{\mbox{\scriptsize pa}(i)}\nonumber\\
&=& \sum_{k=1}^{\infty}\sum_{r=0}^{\lfloor \frac{k}{2}\rfloor} f_y^{(k)}(z_{0})\frac{1}{k!}{k\choose 2r}\times  (2r-1)\ldots 3.1\times[(1-\rho ^2)]^r\nonumber\\& \times & \sum_{s=0}^{k-2r}{k-2r\choose s} (-z_0 +\beta_i z_i)^{s} \int(\beta_{{\mbox{\scriptsize pa}(i)}}^{T}z_{{\mbox{\scriptsize pa}(i)}})^{k-2r-s} P(z_{\mbox{\scriptsize pa}(i)})\ d(z_{\mbox{\scriptsize pa}(i)})
\nonumber\\
&=& \sum_{k=1}^{\infty}\sum_{r=0}^{\lfloor \frac{k}{2}\rfloor} f_y^{(k)}(z_{0})\frac{1}{k!}{k\choose 2r}\times  (2r-1)\ldots 3.1\times[(1-\rho ^2)]^r \nonumber \hspace{-3mm}\\  &\times & \sum_{s=0}^{k-2r}{k-2r\choose s} (-z_0 +\beta_i z_i)^{s}E[(\beta_{x_{\mbox{\scriptsize pa}(i)}}^{T} Z_{{\mbox{\scriptsize pa}(i)}})^{k-2r-s}]
\label{eqn:Intervention effect for nonparanormal}
\end{eqnarray}
We get the following expression for the causal effect,
\begin{eqnarray}
{\frac{\partial}{\partial x_i}}E[Y|{\doop(X_i = \it{x_i})}] 
&=&{\frac{\partial}{\partial z_i}}E[Y|{\doop(X_i = x_i)}]{\frac{\partial z_i}{\partial x_i}}
\label{eqn:drivitiveCE}
\end{eqnarray}
%\textcolor{blue}{
where ${\frac{\partial z_i}{\partial x_i}}=(f_i^{-1})^{\prime}(x_i)$. Therefore, with plugging (\ref{eqn:Intervention effect for nonparanormal}) into (\ref{eqn:drivitiveCE}) proof is completes .%}
\end{proof}
%\textcolor{blue}{
 We  have obtained the general expression (\ref{eq:ce}) for a nonparanormal causal effect. The value of this theorem is that it gives us insight in how higher order moments of the effect $Y$, captured in the higher order derivatives of $f_y$, affect the causal effect, whereas higher order moments of the cause $X_i$ do not.  In practice, this formula is not very helpful as it contains information about the system that we typically do not possess, such as the correlation structure of the latent normal variable. However, this formula can inspire practical estimation procedures of the causal effects in nonparanormal systems. Whereas this is in principle possible, we restrict our attention in this paper to a lower order Taylor approximations in section (\ref{sec:NCE}), since they tend to be more stable. 

%%%%%%%%%%%%%%%%%%%%%%%%%%%%%%%%%%%%%%%%%%%%%%%%%%%%%%%%%%%%%%%%%%%%%%%%%%%%%%%
\subsection{Special case}
\label{sec:Special Cases of Theorem}
%%%%%%%%%%%%%%%%%%%%%%%%%%%%%%%%%%%%%%%%%%%%%%%%%%%%%%%%%%%%%%%%%%%%%%%%%%%%%%% 
We consider the special case of the above theorem for the situation that only $Y$ is normally distributed, and the $X_i$s are still nonparanormal. 

\begin{corollary}
Let $(X_1, \ldots,X_{p-1})\thicksim \mbox{NPN}(0,\Sigma,f)$ and $f_i$ $(i=1,\ldots,p-1)$ is differentiable and $Y\sim N(\mu,\sigma^2)$, then the causal effect of $X_i$ on $Y$ in causal graph $G$ is given by
\begin{eqnarray}
%{\frac{\partial}{\partial y_i}}E[Y_j|{\doop(Y_i = y_i)}]
\mbox{CE}(Y| X_i = x_i) &=& \sigma \beta_i (f_i^{-1})^{\prime}(x_i) 
\label{eqn:Causal effect for nonparanormal linear}
\end{eqnarray}
where $\beta_i$ is defined as in Theorem \ref{thm:npnCE}. 
\end{corollary}
The result simply follows from $f_y(Y)=\mu+\sigma Z$ for $Z$ standard normal. This special case both inspires an estimator for the causal effect and gives some hope for obtaining some consistency results. 

%%%%%%%%%%%%%%%%%%%%%%%%%%%%%%%%%%%%%%%%%%%%%%%%%%%%%%%%%%%%%%%%%%%%%%%%%%%%%
\section{NCE:  nonparanormal causal effect estimator}
\label{sec:NCE}

In this section, we propose a simple estimator for the causal effect that is able  to capture non-linear  effects for a wide ranging collection of distributions. Furthermore, we show that under some conditions, this estimator is consistent.

%%%%%%%%%%%%%%%%%%%%%%%%%%%%%%%%%%%%%%%%%%%%%%%%%%%%%%%%%%%%%%%%%%%%%%%%%%%%%
% we obtained the general expression for a nonparanormal causal effect (\ref{eq:ce}). In practice, this formula is not very helpful as it contains information about the system that we typically do not possess, such as the distribution of the latent normal variable. However, this formula can inspire practical estimation procedures of the causal effects in nonparanormal systems. Whereas this is in principle possible, we restrict our attention in this paper to even lower order Taylor approximations, since they tend to be more stable. 
%
%If we assume that the underlying function $f$ can be appropriately be described by a cubic spline, then in terms of estimation, the terms $f^{(k)}$ can be set to zero for $k\geq 4$. This would reduce the infinite sum in (\ref{eq:ce}) to a sum of merely four terms. These terms, however, still require some estimates of $\rho$ and the various moments of $Z_{\mbox{\scriptsize pa}(i)}$. 

%%%%%%%%%%%%%%%%%%%%%%%%%%%%%%%%%%%%%%%%%%%%%%%%%%%%%%%%%%%%%%%%%%%%%%%%%%%%%%%
\subsection{First order estimator}
\label{sec:First order estimator}
%%%%%%%%%%%%%%%%%%%%%%%%%%%%%%%%%%%%%%%%%%%%%%%%%%%%%%%%%%%%%%%%%%%%%%%%%%%%%%% 
%\textcolor{blue}{
\begin{figure}[tb]
   \centering 
     \hspace{2cm}\includegraphics[scale=1.1,angle=0]{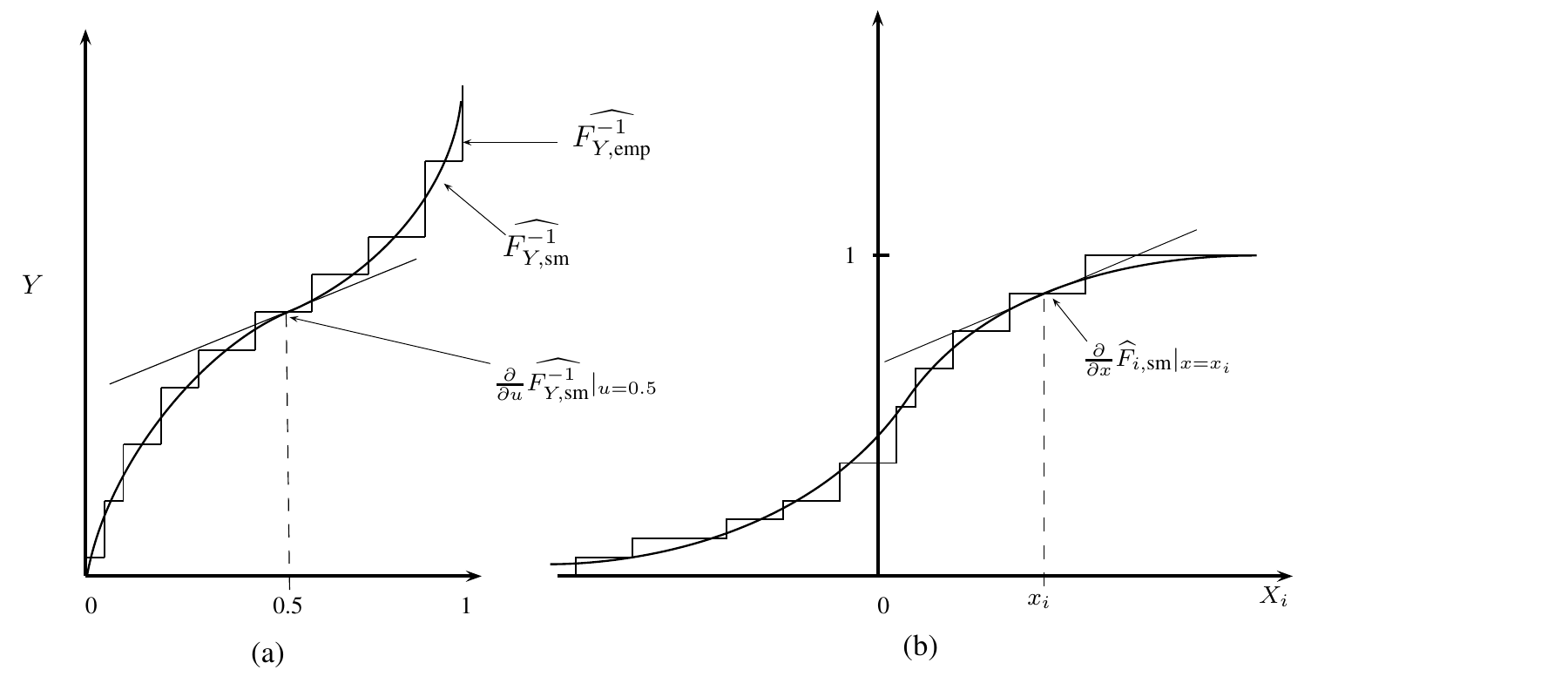}
    
  %}
\caption{(a) the derivative of monotone increasing spline $ \widehat{F_{Y,\mbox{\scriptsize sm}}^{-1}}$ for estimate $\frac{\partial}{\partial x} F_Y^{-1}$. (b) the derivative of the monotone increasing estimating spline $\widehat{F}_{i,\mbox{\scriptsize sm}}$ for estimate $\frac{\partial}{\partial x} F_{i}$.   } % and 20 random
  %variable orderings per graph
  \label{fig:newplot}
\end{figure}

In the special case of the general causal effect theorem, we derived a one term expression that can be used as inspiration for a first order Taylor estimator of the general causal effect of $X_i=x$ on $Y$, i.e.,
%We summarize that the intervention distribution in (\ref{eqn:Intervention1}) and (\ref{eqn:Intervention}) can be inferred
%from the observational distribution $ P$ and the corresponding DAG $G$. All what we
%require is the Markov condition of $P$ with respect to $G$. Furthermore, from (\ref{eqn:causal effect for nonparanormal}) we
%see that each causal effect can be inferred from a local property of the DAG $G$ (namely
%the nodes corresponding to the variables $Z_i$ and ${\mbox{\scriptsize pa}(i)}$ on $Z$) and derivative of functions  of nonparanormal Density and  inverse nanparanormal density. In this paper, we just confront with the $k=1$ in the Taylor expansion and for another value of $K$ still can be discovered. 
%\begin{eqnarray}
%E(Y|\doop(X_i = \it{x_i}))&=&
%\int E(Y|x_i,x_{\mbox{\scriptsize pa}(i)})f(x_{\mbox{\scriptsize pa}(i)})\ d(x_{\mbox{\scriptsize pa}(i)})\nonumber \\
%%&\cong & \int[C_1+ f_j^{,}(z_{0j})\beta_i f_i^{-1}(y_i)]f(pa_i)\ d(pa_i)\\
%&\approx&C+ f_y^{\prime}(0)\beta_i f_i^{-1}(x_i)
%\label{eqn:Intervention effect for nonparanormal linear}
%\end{eqnarray}
%where $C$ is some constant.
%From (\ref{eqn:Intervention effect for  nonparanormal linear}), %}
\begin{eqnarray}
\widehat{\mbox{NCE}}_{z_0}(x) &= & \hat{f}_y^{\prime}(z_0)~{\hat{\beta}}_i ~(\widehat{{f}_i^{-1}})^ {\prime}(x),
\label{eqn:Causal effect for nonparanormal linear} 
\end{eqnarray}
for some $z_0, x\in \mathbb{R}$ and where $\hat{\beta} _i $ is the linear regression  coefficient of $\widehat{f_y^{-1}}(Y)$ on $\widehat{f_i^{-1}}(X_i)$, while controlling for the parents $\widehat{f_{\mbox{\scriptsize pa}(i)}^{-1}}({X_{\mbox{\scriptsize pa}(i)}})$ of $i$. In order to obtain consistency, we trim the data for each variable below its $\alpha/p$ and above  $1-\alpha/p$ quantiles, where $p$ is the number of random variables $(X,Y)$.  When an observation has been trimmed for one variable, it is removed in its entirety for all variables. This means that in the worst case scenario, $1-2\alpha$ of the observations remain. In practice, we will often use $\alpha=0.05$. 

We can simplify  expression  (\ref{eqn:Causal effect for nonparanormal linear}) by considering the case that $z_0=0$. Note that it is straightforward to obtain  
\begin{eqnarray}
\  f^{'}_y(0)
%&=&\frac{\partial}{\partial z}( F_Y^{-1}\circ \Phi(z))|_{z=0}\nonumber\\
%&=&\frac{\partial}{\partial u} F_Y^{-1}(u)|_{u=\Phi(z)=0.5}\ \frac{\partial}{\partial z}\Phi(z)|_{z=0}\nonumber\\
&=&\frac{\partial}{\partial u} F_Y^{-1}(u)|_{u=0.5}\ \phi(0)\nonumber
\label{eqn:f prime in ziro} \\
({f}_i^{-1})^ {\prime}(x)
&=&\left[\phi( f_i^{-1}(x))\right]^{-1}\ \frac{\partial}{\partial x} F_i(x)\nonumber
\label{eqn:f prime in ziro} 
\end{eqnarray}
where $\phi$ is the density function of a standard normal distribution. Considering  Figure \ref{fig:newplot}, $F_Y^{-1}$ will be estimated via a monotone increasing smoother $ \widehat{F_{Y,\mbox{\scriptsize sm}}^{-1}}$, which gives us direct access to its derivative. Similarly, $\frac{\partial}{\partial x} F_{i}$ will be estimated by taking the derivative of the monotone increasing estimating smoother $\widehat{F}_{i,\mbox{\scriptsize sm}}$. In particularly, we will make use of kernel smoothers, as explained in the next section in order to prove consistency. Finally, $f_i^{-1}(x)$ will be estimated as $\hat{z} = \Phi^{-1}(\hat{F}_{i,\mbox{\scriptsize sm}}(x))$. Putting this together, we obtain a simplified and 
%
% So we propose the estimator,
%\begin{equation}
% \hat f^{'}_y(0)=\phi(0) ~\frac{\partial}{\partial u}\widehat{F_{Y,\mbox{\scriptsize spl}}^{-1}}|_{u=0.5}.
% \label{eq:dfy0}
% \end{equation}
%%Similarly, for$(\hat{f}_i^{-1})^ {\prime}(x_i)$ we obtain 
%%\begin{eqnarray}
%%({f}_i^{-1})^ {\prime}(x_i)&=&\frac{\partial}{\partial x}( \Phi^{-1}\circ  F_i)|_{x=x_i}\nonumber\\
%%&=&\frac{\partial}{\partial u}\Phi^{-1}|_{u=F(x_i)}~ \frac{\partial}{\partial x} F_i|_{x=x_i}\nonumber\\
%%&=&\left[\frac{\partial}{\partial z}\Phi|_{z=\Phi^{-1}\circ  F_i(x_i)}\right]^{-1}\ \frac{\partial}{\partial x} F_i|_{x=x_i}\nonumber\\
%%&=&\left[\phi( f_i^{-1}(x_i))\right]^{-1}\ \frac{\partial}{\partial x} F_i|_{x=x_i}\nonumber
%%\label{eqn:f prime in ziro} 
%%\end{eqnarray}
%where  as shown in Figure \ref{fig:newplot}(b). Then we propose the estimator, 
%\begin{equation}
%(\hat{f}_i^{-1})^ {\prime}(x_i)=\left[\phi(z_i)\right]^{-1} \frac{\partial \widehat{F}_{i,\mbox{\scriptsize spl}}}{\partial x}(x_i),
%\label{eq:dfix}
%\end{equation}
%where $z_i$ is the standard normal quantile corresponding to the rank of $x_i$. Putting (\ref{eqn:Causal effect for nonparanormal linear} ), (\ref{eq:dfy0}) and (\ref{eq:dfix}) together, we obtain an 
explicit estimator of a non-paranormal causal effect, 
\begin{eqnarray}
\widehat{\mbox{NCE}}_0(x) &= &  \hat{\beta}_i ~ \frac{\phi(0)}{\phi(\hat{z})} ~\frac{\partial \widehat{F_{Y,\mbox{\scriptsize sm}}^{-1}}}{\partial u}(0.5) ~ \frac{\partial \widehat{F}_{i,\mbox{\scriptsize sm}}}{\partial x}(x) .
\label{eq:ce-est} 
\end{eqnarray}
In the following section, we will show that under certain conditions the above estimator is consistent.

%%%%%%%%%%%%%%%%%%%%%%%%%%%%%%%%%%%%%%%%%%%%%%%%%%%%%%%%%%%%%%%%%%%%%%%%%%%%%%%
%\subsection{Higher Order Estimators}
%\label{sec:Higher order estimators}
%%%%%%%%%%%%%%%%%%%%%%%%%%%%%%%%%%%%%%%%%%%%%%%%%%%%%%%%%%%%%%%%%%%%%%%%%%%%%%%  

%%%%%%%%%%%%%%%%%%%%%%%%%%%%%%%%%%%%%%%%%%%%%%%%%%%%%%%%%%%%%%%%%%%%%%%%%%%%%%%
\subsection{Consistency}
\label{sec:Consistency}
%%%%%%%%%%%%%%%%%%%%%%%%%%%%%%%%%%%%%%%%%%%%%%%%%%%%%%%%%%%%%%%%%%%%%%%%%%%%%%% 
In this section we will be concerned with the asymptotic
behaviour of our estimator in (\ref{eq:ce-est})  under the assumption of normality of $Y$. 
We first show that the random, but not necessarily independent, sampling scheme of $(X_1,\ldots,X_{p-1})\sim NPN(0,\Sigma,f)$ and $Y\sim N(\mu,\sigma^2)$ combined with our lower and upper $\alpha/p$ trimming scheme will eventually fill up the $p$-dimensional cube $[L_\alpha,U_\alpha]$, where $L_\alpha = (L_\alpha^1,\ldots,L_\alpha^{p-1},L_\alpha^y)$ and  $U_\alpha= (U_\alpha^1,\ldots,U_\alpha^{p-1},U_\alpha^y)$ are the lower and upper quantiles, respectively, for each of the variables $(X_1,\ldots,X_{p-1},Y)$. From the original sample size $n$ approximately $(1-2\alpha)n$ will fall in this cube. Then we show that the kernel estimators of the functions used in the NCE estimators and their derivatives converge fast to their true values in probability. Together with the fact that products of consistent estimators are consistent, this proves the consistency of the  estimator $\widehat{\mbox{NCE}}_0(x)$.

\begin{proposition}
\label{cndkernel1} 
Consider any absolutely continuous random variable $X$ with lower $\alpha$ quantile $L_\alpha$ and upper $\alpha$ quantile $U_\alpha$. For the $N\asymp (1-2\alpha)n$ ordered observations of $X$ in the finite interval $[L_\alpha, U_\alpha]$, the following property holds
\[
\max_{2\leq i\leq N} |X_{(i)}-X_{(i-1)} | = O_P(1/N). 
\]
\end{proposition}
\noindent The symbol $\asymp$  denotes that two sequences of real numbers are asymptotically of the same order. The proof of this Proposition is a simple exercise and will not be given here. 
 
Our goal is first to estimate the function  $ F_{i}$ and its derivative $\frac{\partial}{\partial x} F_{i}$. Similarity, we aim to estimate  $F_{i}^{-1}$ and its derivative.
In order to derive asymptotic properties, we will be using kernel estimators for $\hat{F}_{i,sm}$ and $\widehat{F_{i,sm}^{-1}}(x)$, respectively,
\begin{eqnarray}
\hat F_{i,n}(x)&=&\sum_{j= 2}^N(x_{i(j)}-x_{i(j-1)})\frac{1}{b_n}K\left(\frac{x-x_{i(j)}}{b_n}\right)\left(\alpha+\frac{j-1}{n}\right)\label{eq:kernel1}\\ 
\widehat{F_{i,n}^{-1}}(u)&=&\sum_{j=1}^{N}\frac{1-2\alpha}{N}\frac{1}{b_n}K\left(\frac{u-(\alpha+\frac{j(1-2\alpha)}{N})}{b_n}\right)x_{i(j)}
\label{eq:kernel2}
\end{eqnarray}
for $x\in [L_\alpha^i,U_\alpha^i]$ and $u\in [\alpha, 1-\alpha]$, where $K$ is a kernel function, $b_n > 0$ denotes the bandwidth that we take to
depend on the sample size $n$ in such a way that $b_ n \rightarrow 0$ as $n \rightarrow \infty$ and $x_{i(1)}, x_{i(2)}, \ldots, x_{i(N)}$ denote the order statistics of that part that for the $i$ variable that falls within $[L_\alpha^i,U_\alpha^i]$. 
We define an estimator of $\frac{\partial}{\partial x} F_{i}$  by taking the derivative of the kernel smoother $\widehat{\frac{\partial}{\partial x}  F_{i,n}} = {\frac{\partial}{\partial x} \hat F_{i,n}}=\hat F'_{i,n}$.

\begin{proposition}
\label{priestleyprop} 
If the kernel $K$ is symmetric and twice
continuously differentiable with support in $[-1,1],$ and if
it satisfies the integrability conditions (a) $\int_{-1}^1 K(u)\,\mathrm{d}u=1$
and (b) $\int_{-1}^1 u^{\ell}K(u)\,\mathrm{d}u=0$ for $\ell=1,\ldots,\gamma-1$, then 
for a fixed number $\delta,$ such that $\alpha<\delta<1/2:$
\item[(i)] If $F$ and $F^{-1}$ are  $\gamma\geq1$ times continuously differentiable
and $b_n\rightarrow0$ as $n\rightarrow\infty,$ then
%
%e3.4 ###
\begin{eqnarray*}
\sup_{x\in[L_\alpha^i, U_\alpha^i]}|{\hat F_{i,n}}(x)- F_{i}(x)| &=& 
\mathrm{O}_P \Biggl(
b_n^{\gamma}+\frac{1}{nb_n^2}+\sqrt{\frac{\log n}{nb_n}}
 \Biggr).\\
  \sup_{u\in[\delta,1- \delta]}|\widehat{F_{i,n}^{-1}}(u)- F^{-1}_{i}(u)| &=&
\mathrm{O}_P \Biggl(
b_n^{\gamma}+\frac{1}{nb_n^2}+\sqrt{\frac{\log n}{nb_n}}
 \Biggr).  
\end{eqnarray*}

\item[(ii)] If $F$ and $F^{-1}$ are $\gamma\geq2$ times continuously differentiable and
$b_n\rightarrow0$ as $n\rightarrow\infty,$
then

%
%e3.5 ###
\begin{eqnarray*}
\label{muprimep}
\sup_{x\in[L_\alpha^i, U_\alpha^i]}|{\hat F^{\prime}_{i,n}}(x)- F'_{i}(x)|
&=& \mathrm{O}_P \Biggl(
b_n^{\gamma-1}+\frac{1}{nb_n^3}+\sqrt{\frac{\log
n}{nb_n^3}} \Biggr).\\
\sup_{u\in[\delta,1-\delta]}|\widehat{F^{-1}_{i,n}}^{\prime}(u)- {F^{-1}_{i}}^{\prime}(u)|
&=& \mathrm{O}_P \Biggl(
b_n^{\gamma-1}+\frac{1}{nb_n^3}+\sqrt{\frac{\log
n}{nb_n^3}} \Biggr).
\end{eqnarray*}
In particular, $\hat F_{i,n}(x)$ and $\hat F_{i,n}^{\prime}(x)$ are consistent on  $[L_\alpha^i, U_\alpha^i]$ and 
$\widehat{F_{i,n}^{-1}}(x)$ and $\widehat{F^{-1}_{i,n}}^{\prime}(x)$
are consistent on $[\delta,1-\delta],$ if $nb_n^3/\log
n\rightarrow\infty$ holds additionally.
\end{proposition}
The proof is given in \citet[Proposition 3.1]{gugushvili2012}. The estimator $\widehat{\mbox{NCE}}_0(x)$ in (\ref{eq:ce-est}) contains four terms. Based on Proposition \ref{priestleyprop} we showed the consistency of two terms, $\hat F_{i,n}^{\prime}(x)$ and $\widehat{F^{-1}_{i,n}}^{\prime}(x)$. As any continuous function of a consistent estimator is consistent \citep{lehmann1999elements}, also $\hat{z} = \Phi^{-1}(\hat{F}_{i,\mbox{\scriptsize n}}(x))$  is consistent. In order to proof consistency of $\widehat{\mbox{NCE}}_0(x)$ we still  need to show that $\hat{\beta}_i$ is consistent, where $\hat{\beta} _i $ is the linear regression  coefficient of $\widehat{f_y^{-1}}(Y)$ on $\widehat{f_i^{-1}}(X_i)$, while controlling for the parents $\widehat{f_{\mbox{\scriptsize pa}(i)}^{-1}}({X_{\mbox{\scriptsize pa}(i)}})$ of $i$. In the following Proposition we show consistency of $\hat{\beta}_i$. 

%###############################################################
%TO DO
%\begin{enumerate}
%\item DEFINE $\widehat{F_{i,\mbox{\scriptsize{sm}}}^{-1}}$ and its derivative (for $i\leq p-1$ or $y$) and show that $z_n = \Phi(\widehat{F_{i,n}^{-1}}(x)) 
%\stackrel{P}{\longrightarrow} \Phi(F^{-1}_{i}(x))$
%\item DEFINE $\hat{\beta}_i$ and show its consistency. Use the proof of the proposition 3. 
%\end{enumerate}
\begin{proposition}
\label{betaasy}
Let  
%$\hat{z} = \Phi^{-1}(\hat{F}_{i,\mbox{\scriptsize n}}(x))$ and 
$\hat{\beta} _i $ be the linear regression  coefficient of $\widehat{f_y^{-1}}(Y)$ on $\widehat{f_i^{-1}}(X_i)$, while controlling for the parents $\widehat{f_{\mbox{\scriptsize pa}(i)}^{-1}}({X_{\mbox{\scriptsize pa}(i)}})$ of $i$, then
\begin{eqnarray}
\label{beta}
\hat\beta_i^n \stackrel{P}{\longrightarrow} \beta_i,
\end{eqnarray}
where $\beta_i$ is the true regression coefficient as defined in Theorem \ref{thm:npnCE}. 
\end{proposition}
\begin{proof}
Define $$ \hat Z_{n} = 
 \begin{pmatrix}
  \hat z_{1,i} & \hat z_{1,\mbox{\scriptsize pa}(i)_1} & \cdots & \hat z_{1,\mbox{\scriptsize pa}(i)_k} \\
  \hat z_{2,i} & \hat z_{2,\mbox{\scriptsize pa}(i)_1} & \cdots & \hat z_{2,\mbox{\scriptsize pa}(i)_k} \\
  \vdots  & \vdots  & \ddots & \vdots  \\
  \hat z_{N,i} & \hat z_{N,\mbox{\scriptsize pa}(i)_1} & \cdots & \hat z_{N,\mbox{\scriptsize pa}(i)_k} 
 \end{pmatrix},$$
such that $\hat{z}_{j,l} = \Phi^{-1}(\hat{F}_{l,\mbox{\scriptsize n}}(x_{jl}))$ where $x_{jl}$ is the non-ordered $j$th sample of variable $l$ and $\mbox{pa}(i)$ is the index set of $k$ parents of $i$. Let
$ \hat \Upsilon_{n}^{T} =\left(
  \Phi^{-1}(\hat{F}_{y,\mbox{\scriptsize n}}(y_1)), \Phi^{-1}(\hat{F}_{y,\mbox{\scriptsize n}}(y_2)), \cdots, \Phi^{-1}(\hat{F}_{y,\mbox{\scriptsize n}}(y_N))\right)$. The coefficient $\hat{\beta}_i^n$ is defined as the first element of the vector,
$$\hat{\beta}^n=(\hat Z_n^t\hat Z_n)^{-1}\hat Z_n^t\hat \Upsilon_n.$$ 
We can also define the oracle estimator $\hat{B}_i^n$ as the first element of
 $$\hat{B}^n=( Z_n^t Z_n)^{-1} Z^t \Upsilon_n,$$ 
where $Z_n$  and $\Upsilon_n$ are obtained by replacing the marginal $\hat F$s by the true $F$s. Consider an arbitrary $\epsilon,\delta>0$, 
\begin{eqnarray}
P(|\hat \beta_i^n- \beta_i|>\epsilon)&=&P(|\hat \beta_i^n-\hat{B}_i^n+\hat{B}_i^n- \beta_i|>\epsilon) \nonumber\\
&\leq&P((|\hat \beta_i^n-\hat{B}_i^n|+|\hat{B}_i^n- \beta_i|)>\epsilon) \nonumber \\
&\leq&P((|\hat \beta_i^n-\hat{B}_i^n|> \epsilon/2) +  P(|\hat{B}_i^n- \beta_i|)>\epsilon/2) \label{eq:RHS0}
\end{eqnarray}
We first consider the first right hand side term of (\ref{eq:RHS0}). Let's define $\hat A_n=\frac{\hat Z_n^t\hat Z_n}{n}$,  $ {A_n}= \frac{ Z_n^t Z_n}{n}$ and $\hat b_n=\frac{\hat Z_n^t\hat \Upsilon_n}{n}$ and $ {b}_n= \frac{ Z_n^t \Upsilon_n}{n}$. Then,
\begin{eqnarray}
P(|\hat \beta_i^n-\hat{B}_i^n|>\frac{\epsilon}{2})&\leq&P(\parallel\hat A_n^{-1}\hat b_n-{A}_n^{-1}{b}_n\parallel ^2>\frac{\epsilon}{2}) \nonumber \\
&\leq& %P(\parallel\hat A_n^{-1}\hat b_n-\hat{A}_n^{-1}{b}_n\parallel ^2+\parallel\hat A_n^{-1}{b} _n-{A}_n^{-1}{b}_n\parallel ^2>\frac{\epsilon}{2})\\ &=&
P(\parallel\hat A_n^{-1}(\hat b_n-{b}_n)\parallel ^2+\parallel(\hat A_n^{-1}-{A}_n^{-1}){b}_n\parallel ^2>\frac{\epsilon}{2})\nonumber\\
&\leq&P(\parallel\hat A_n^{-1}(\hat b_n-{b}_n)\parallel ^2> \frac{\epsilon}{4})+P(\parallel(\hat A_n^{-1}-{A}_n^{-1}){b}_n\parallel ^2>\frac{\epsilon}{4}) \label{eq:RHS1}.
\end{eqnarray}
By the consistency of $\hat{z}$, we have that both $\hat{b}_n$ and $b_n$ converge in probability to some $b=\Sigma_{(i,\mbox{\scriptsize pa}(i)),p}$ and both $\hat{A}^{-1}_n$ and $A_n^{-1}$ converge in probability to some $A^{-1}=\Sigma_{(i,\mbox{\scriptsize pa}(i)),(i,\mbox{\scriptsize pa}(i))}^{-1}$, where $\Sigma$ is defined in the body of Theorem \ref{thm:npnCE}. Therefore, there is a $n^\ast$, such that for all $n\geq n^\ast$, both terms on the right hand side of (\ref{eq:RHS1}) are less than $\delta/4$.  So for all $n\geq n^\ast$,
 \[ P(|\hat \beta_i^n-\hat{B}_i^n|>\frac{\epsilon}{2}) <  \frac{\delta}{2}. \]
For the second term of the right hand side of (\ref{eq:RHS0}), it is sufficient to use the fact that in the latent normal space a regression estimate is consistent and therefore, there exist a $n^\perp$, such that any $n>n^\perp$, 
$$P(|\hat{B}_i^n- \beta_i|>\epsilon/2)<\delta/2.$$
Putting both results together, we now have that for any $n\geq \max\{n^\ast,n^\perp\}$, 
$$P(|\hat \beta_i^n- \beta_i|>\epsilon)<{\delta}.$$
Thus we get the desired result.
\end{proof}
%\begin{proposition}
%\label{pro2} Let ${Z_n}$ and ${W_n}$ be sequences of random variables on $(\Omega, F, P)$,
%and let ${an}$ and ${bn}$ be sequences of positive real numbers.
%\item[(1)] Suppose $Z_n = o_p(a_n)$ and $W_n = o_p(b_n)$; then,
%(a)$ Z_nW_n = o_p(a_nb_n)$.
%%(b) $X_n + Y_n = o_p(max(a_n, b_n))$.
%\item[(2)] If $Z_n = o_p(a_n)$ and $W_n = O_p(b_n)$, then $Z_nW_n = o_p(a_nb_n)$.
%\end{proposition}

%For showing the consistency of $z_n = \Phi(\widehat{F_{i,n}^{-1}}(x)) 
%\stackrel{P}{\longrightarrow} \Phi(F^{-1}_{i}(x))$ we only apply  Lemma \ref{lemma1} part (2).  In Proposition \ref{consistency}  we also use Lemma \ref{lemma1} part (1) about products of consistent estimators.

%\begin{lemma}
%\label{lemma1}\item[(1)]
%\item[(2)] Let $f : \mathbb{R}\rightarrow\mathbb{R}$ be a continuous function. If $Z_n$ converges in probability to $Z$, then  $f(Z_n)$ converges
%in probability to $f(Z)$.
%\end{lemma}
The following Proposition  provides a result that our estimator in (\ref{eq:ce-est}) is consistent. %converges to (\ref{eqn:Causal effect for nonparanormal linear}).
\begin{proposition}
\label{consistency} Consider the estimator of NCE$_0(x)$ in (\ref{eq:ce-est}), for which we consider the component estimators (\ref{eq:kernel1}),  (\ref{eq:kernel2}) and (\ref{beta}). For the kernel estimators, we assume that the conditions of Proposition \ref{priestleyprop} are satisfied and, furthermore, the bandwidth $b_n\rightarrow0$, but not too fast so that $nb_n^3/\log n\rightarrow\infty$. We have $$\widehat{\mbox{NCE}}_{0,n}\convp\mbox{NCE}_{0}.$$
\end{proposition}
\begin{proof}
For two sequences of random variables $Z_n$ and $W_n$ and two random variables $Z, W$, such that  $Z_n$ converges in probability to $Z$ and $W_n$ converges in probability to $W$, then it is a standard results that $Z_nW_n$ converges in probability to $ZW$ \citep{lehmann1999elements}. As all the components of $\widehat{\mbox{NCE}}_0(x)$ have been shown to be consistent, then the estimator is consistent. 
\end{proof}

%%%%%%%%%%%%%%%%%%%%%%%%%%%%%%%%%%%%%%%%%%%%%%%%%%%%%%%%%%%%%%%%%%%%%%%%%%%%%%%
\section{Simulation studies}
\label{sec:simulation-methods}
%%%%%%%%%%%%%%%%%%%%%%%%%%%%%%%%%%%%%%%%%%%%%%%%%%%%%%%%%%%%%%%%%%%%%%%%%%%%%%%
In this section, we test our estimation method for two different types of distributions, to wit, Gaussian and nonparanormal with exponential margins. For Gaussian data, the 
method should find constant causal effects and can be compared directly with the IDA method \citep{maathuis2009estimating}. We consider two scenarios: (i) in which the underlying causal graph is known and (ii) where it is unknown and needs to be estimated via the RPC-algorithm. In the latter case, the IDA method has some additional advantages of being able to use the somewhat more powerful PC-algorithm. 
For the nonparanormal simulation with exponential margins, calculating the explicit causal effect is very involved in general. Therefore we apply the method to a network with two nodes for which the true causal effect can be evaluated numerically.
%
%\begin{figure}[!h]\centering
%  \centering 
%  %\subfloat[True DAG]{\includegraphics[width=0.4\textwidth]{Tdag.pdf}\label{fig:f1}}
%  \subfigure[True DAG]{\includegraphics[width=0.4\textwidth]{Tdag.pdf}\label{fig:f1}}
%  % \subfloat[Simulate DAG]{\includegraphics[width=0.4\textwidth]{Sdag.pdf}\label{fig:f2}}
% \subfigure[CPDAG]{\includegraphics[width=0.32\textwidth]{CPDAG8.pdf}\label{fig:f3}}
%  \subfigure[One of three DAGs consistent with CPDAG]{\includegraphics[width=0.4\textwidth]{Sdag.pdf}\label{fig:f2}}
% 
%  \caption{Plots generated using the R-package {\tt pcalg}.  (a) The true DAG.
%  (b) The estimated CPDAG using the PC algorithm, based on a a simulated dataset with $n=1,000$ replicates and significance cut-off $\alpha=0.01$. 
%(c) One of the three DAGs consistent with CPDAG.}
%  \label{fig:smallgraph}
%\end{figure}

%%%%%%%%%%%%%%%%%%%%%%%%%%%%%%%%%%%%%%%%%%%%%%%%%%%%%%%%%%%%%%%%%%%%%%%%%%%%%%%
\subsection{Gaussian data}
\label{sec:simulation-results}

Following  \cite{kalisch2007estimating}, we simulate random DAGs and sample from probability distributions faithful to them. For convenience, we fix an increasing ordering of the variables $\{X_1, . . . , X_p\}$, meaning that for  a vector of independent Gaussian variables $\varepsilon = (\varepsilon_1 , . . . , \varepsilon_p )$
\begin{equation}
X= A X+\varepsilon,
\label{eqn:simu matrix}
\end{equation}
where the coefficient matrix $A$  has entries $A_{ij}$ that are zero for $i<j$ and $A_{ji} \neq 0$ if the corresponding DAG has a directed edge from node $i$ to node $j$ for some $i>j$. The DAGs and skeletons thereof
that are created in this way have the following property: $ E[N_i ] = s(p - 1)$, where $N_i$ is the number of neighbours of a node $i$.
With probability one, the vector X solving (\ref{eqn:simu matrix}) is Markov and faithful with respect to G.

We consider two different size graphs: a small graph with ten vertices and a larger graph with  fifty vertices, both with
an expected vertex degree of three. For each $n \in \{100,1000\}$ and each of the two types of graphs, we repeat each simulation 100 times.

%%%%%%%
\subsubsection{Causal DAG known}
\label{sec:simulation-results}
%%%%%%%%%%%%%%%%%%%%%%%%%%%%%%%%%%%%%%%%%%%%%%%%%%%%%%%%%%%%%%%%%%%%%%%%%%%%%%%
If we assume that the causal DAG is known, then for estimating the causal effects we apply both our NCE algorithm and the IDA algorithm, described in (\ref{eqn: linear causality for Gaussian}). Given that the IDA algorithm is made for these Gaussian data, the method should outperform the NCE method, which is agnostic about the underlying distributional assumptions. We apply the methods to the four data scenarios and the results are presented in the last column of Table \ref{tab:simu1}. It shows that when the number of observations are increasing, the mean absolute value deviation for causal effect estimates for both IDA and NCE methods are decreasing.  Furthermore, the NCE method, as expected, is more variable. This variation is mostly the result from the poorer estimates of the distributional shape in the tails of the distribution.

%%%%%%%%%%%%%%%%%%%%%%%%%%%%%%%%%%%%%%%%%%%%%%%%%%%%%%%%%%%%%%%%%%%%%%%%%%%%%%%
%%%%%%%%%%%%%%%%%%%%%%%%%%%%%%%%%%%%%%%%%%%%%%%%%%%%%%%%%%%%%%%%%%%%%%%%%%%%%%%
\subsubsection{Causal DAG unknown}
\label{sec:simulation-results}
%%%%%%%%%%%%%%%%%%%%%%%%%%%%%%%%%%%%%%%%%%%%%%%%%%%%%%%%%%%%%%%%%%%%%%%%%%%%%%%
If the underlying causal DAG is considered unknown, then the CPDAG and associated DAGs need to be estimated. 
For each simulation, we run both the standard PC-algorithm and the robust RPC-algorithm on a grid of significance levels $\alpha$ ranging from $10^{-10}$ to 0.5.  
For each estimated DAG,
we compute the causal effects of each node according to the NCE method and the compare the results with the IDA method.
%For illustration purposes Figure \ref{fig:smallgraph} shows the true graph, the estimated CPDAG and one of three equivalent DAGs for the $p=10$ scenario. 

Figures \ref{fig:sim10-2-0-1} show the causal effects between the chosen nodes for small graph on ten vertices, i.e. $p=10$, with $n=100$.  In these figures the red line show the real causal effect between 2  chosen nodes. The blue line shows the average estimated causal effect from the IDA method. The black line show the functional causal effect estimate from (\ref{eq:ce-est}), proposed by our method. The dashed lines express the average standard deviation of our functional causal effect estimate. A clear message emerges from plots: whereas the IDA method is exactly matched for this simulation scenario, our nonparanormal causal effects estimates are quite stable. Moreover, the confidence intervals calculated by our method typically contain the true effect. 
 
%Figure \ref{fig:sim50-2-0-1} and \ref{fig:sim50-1-0-1} also consider a larger graph on  fifty vertices ($p=50$) with $n=100$ and  $n=1000$, respectably.
%In table \ref{tab:simu1} we summarize the all simulation study.
%Results  for comparison our method (NCE) and IDA for  small graph ($p=50$) and large graph ($p=50$) when the data are Gaussian.
In Table \ref{tab:simu1} 
provide numerical comparisons of
both methods on data sets with different transformations, where we repeat the experiments 100
times and report the mean absolute value deviation for causal effect on each pair nodes in both IDA and NCE methods.  
%From the table, the NCE does not achieves smaller errors than the IDA if
%the number of observation increasing, specifically when DAG is known.
%therefore, NCE performance comparable
%to the IDA when the true distribution is exactly multivariate Gaussian.

\begin{table}[h]
\caption{   Results of  mean absolute value deviation causal effect for comparison NCE and IDA methods for  small graph ($p=10$) and large graph($p=50$) when the data is Gaussian.}
\label{tab:simu1}
\begin{center} {\footnotesize}
 \small\addtolength{\tabcolsep}{-5pt}
 \begin{tabular}{cccccccccccccc}
  \hline
   &     &     \multicolumn{2}{c}{ $\alpha=0.01$} && &  & \multicolumn{2}{c}{ $\alpha=0.1$} & & & \multicolumn{2}{c}{ DAG Known}& \\
  \cline{3-4}
  \cline{8-9}
  \cline{12-13}  
 &  \textbf{{ \itshape n}} ~~~~~~~ &\textbf{IDA}& \textbf{ NCE} &&  & &  \textbf{IDA} & \textbf{ NCE} && &  \textbf{IDA} & \textbf{ NCE}&\\
\hline
\textbf{{\itshape p}}= 10   \\
 ~~~~~~~~~~~ & 100 ~~~~~~~ &~  0.101  &~ 0.576 & & &  &0.144&~ 0.554&& &  0.118&0.455 \\[1ex]

   &1000   ~~~~~~~ &~ 0.033 &~ 0.385 & & & &0.029 &~0.283&& & 0.031&0.303 \\[1ex]
    
\textbf{{\itshape p}}= 50   \\
 ~~~~~~~~~~~~~ & 100 ~~~~~~~ &~ 3.732  &~ 2.515 & & &  &2.261 &~ 3.759 && & 2.004 & 2.677 \\[1ex]
		
  	&1000 ~~~~~~~   &~ 1.175 &~ 2.100 & ~~~~~~~~& & & 0.964 &~  1.378& &~~~~~~~~~& 0.724& 2.281 \\[2ex]   
 \hline		
\end{tabular}
\end{center}
\end{table}

\begin{figure}[]
   \centering
     \includegraphics[scale=.7,angle=0]{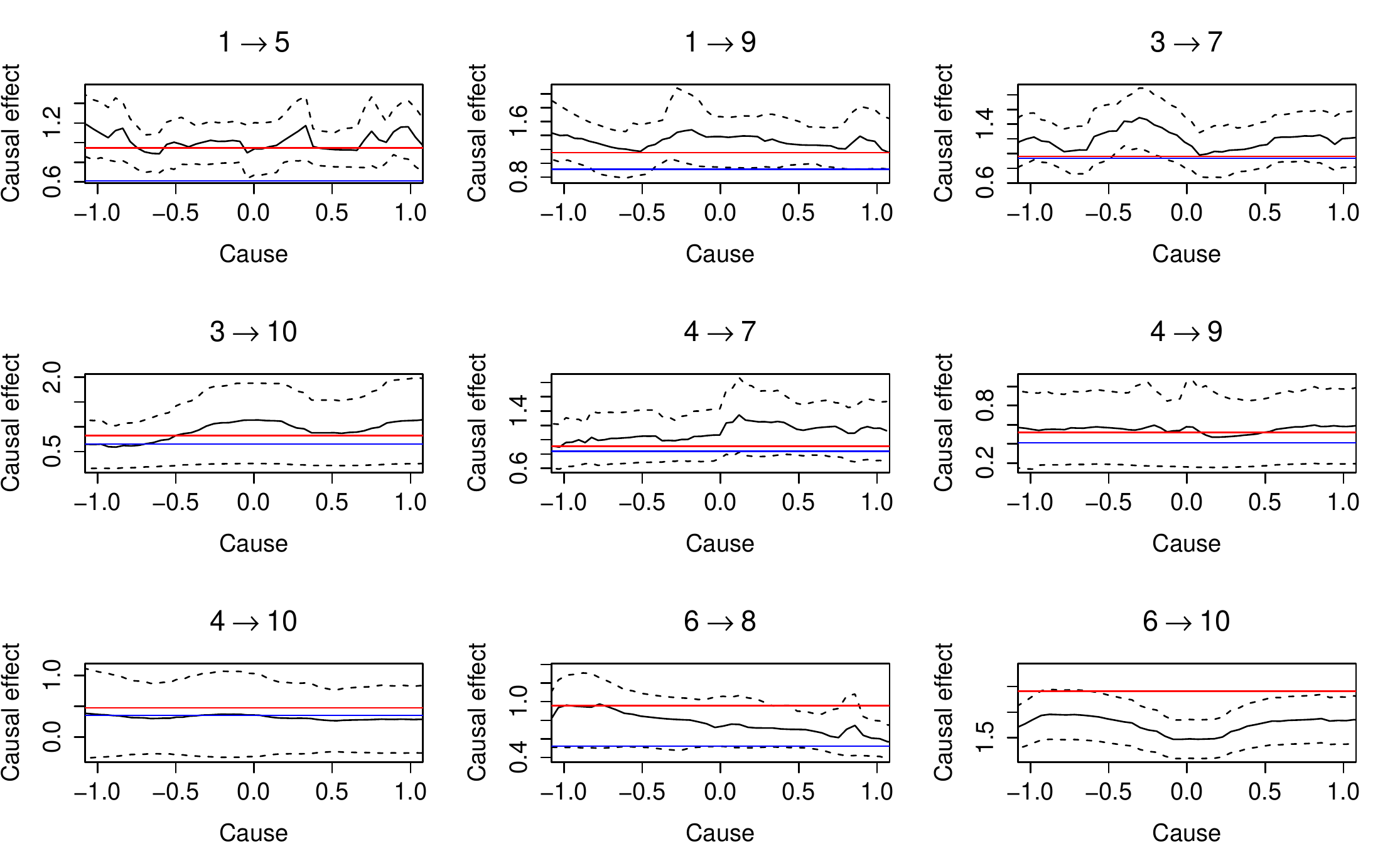}
    
  %}
\caption{Simulation study for Gaussian data from a causal graph (ten vertices $p=10$, with $n=100$ observations). The red lines are the true (constant) causal effects. The blue lines are the causal effect estimates from the IDA methods and black lines show the functional causal effect estimates from our NCE method. The dashed lines show the confidence intervals for functional causal effect estimates.  } % and 20 random
  %variable orderings per graph
  \label{fig:sim10-2-0-1}
\end{figure}
\subsection{Exponential data}
\label{sec:simulation-results}
%%%%%%%%%%%%%%%%%%%%%%%%%%%%%%%%%%%%%%%%%%%%%%%%%%%%%%%%%%%%%%%%%%%%%%%%%%%%%%%

Only in a few special non-Gaussian distributional examples can we calculate the causal effects (\ref{eq:ce}) exactly. This is particularly relevant in a simulation study, where we want to show the efficiency of our estimation method. We consider the causal effects in a bivariate exponential distribution. We assume only two nodes with exponential marginal distributions and then apply \cite{crane2008conditional} to find the closed form for conditional expectation formula for Gaussian copula. We derive the causal effect for the bivariate Gaussian copula. 
%\begin{example}
%\label{exam:exponential}
If we have a bivariate Gaussian copula, with dependence parameter $\rho$, we have
\begin{eqnarray}
 E(Y |X = x) = \int_{\mathbb{R}} y {\frac{\partial}{\partial y}} \Phi(\frac{\Phi^{-1}(F(y))-\rho \Phi^{-1}(G(x))}{\sqrt{1-\rho^ 2}})dy.
 \label{eqn:Gaussian copula}  
\end{eqnarray}
If both marginal distributions $F$ and $G$ were $N(0,1)$, the copula would revert back to the
bivariate normal distribution. The Gaussian copula, however, gives us more flexibility, as it
can accommodate any type of univariate distributions, $F$ and $G$. In (\ref{eqn:Gaussian copula}), we choose two marginal distributions that are exponential with parameter $\lambda_x,\lambda_y >0$. Thus, Equation (\ref{eqn:Gaussian copula}) reduces to
\begin{eqnarray}
 E(Y |X = x)= \frac{1}{\sqrt{1-\rho ^2}}\int_{\mathbb{R}} y \phi(\frac{\Phi^{-1}(1-\exp(\lambda_y y))-\rho \Phi^{-1}(1-\exp(\lambda_x x))}{\sqrt{1-\rho^ 2}}) \frac{{\exp({-\lambda_y y})}}{\phi(\Phi^{-1}(1-\exp(\lambda_y y)))} dy \nonumber 
 \end{eqnarray}
Therefore, for a bivariate nonparanormal with exponential marginals, we obtain the following causal effect,
 \begin{eqnarray}
\mbox{CE}(Y| X = x)= - \frac{\rho}{1-\rho ^2}\int_{\mathbb{R}} y\phi^{'} (t)\frac{{\exp({-\lambda_x x})}{\exp({-\lambda_y y})}}{{\phi(\Phi^{-1}(1-\exp(\lambda_y y)))}{\phi(\Phi^{-1}(1-\exp(\lambda_x x)))}} dy
\label{eqn:Gaussian copula for exp}
\end{eqnarray}
where  $t=\frac{\Phi^{-1}(1-\exp(\lambda_y y))-\rho \Phi^{-1}(1-\exp(\lambda_x x))}{\sqrt{1-\rho^ 2}}$. 

%In the simulation section, we evaluate (\ref{eqn:Gaussian copula for exp}) by numerical integration. See Appendix A for the derivation of (\ref{eqn:Gaussian copula for exp}).
%\end{example}

%We consider the same scenario studied in Example \ref{exam:exponential}. 
In the simulation study we assume that node $X$ affects node $Y$, in the following fashion,
\begin{eqnarray*} X &=& F^{-1}\left(\Phi(Z_1)\right)\\
 Y&=& F^{-1}\left(\Phi(\frac{Z_1+Z_2}{\sqrt{2}})\right),
 \end{eqnarray*}
where $F$ is the CDF of an Exponential(1) distribution and $Z_1,Z_2 \stackrel{\mbox{\scriptsize i.i.d.}}{\sim} N(0,1)$. This falls under the usual nonparanormal scenario. The explicit expression for the causal effect in Theorem \ref{thm:npnCE} is very involved, but we derived in  (\ref{eqn:Gaussian copula for exp}) a simplified expression. We evaluated this expression numerically to obtain the true causal effect, expressed as the solid black line in  Figure \ref{EXP}. Then we simulated $n=1,000$ observations from the above model for inferring the causal effect.

We assume that the underlying causal graph, $X\longrightarrow Y$, is known and used the NCE method to infer the non-linear causal effect. The blue line  Figure \ref{EXP} shows the functional causal effect estimate from NCE method. It matches very well the true causal effect. Clearly, had IDA been applied in this scenario, it would have come up with a nonsensical constant causal effect.

\begin{figure}[!h]\centering%
   % \subfigure[]{
     \includegraphics[scale=.4]{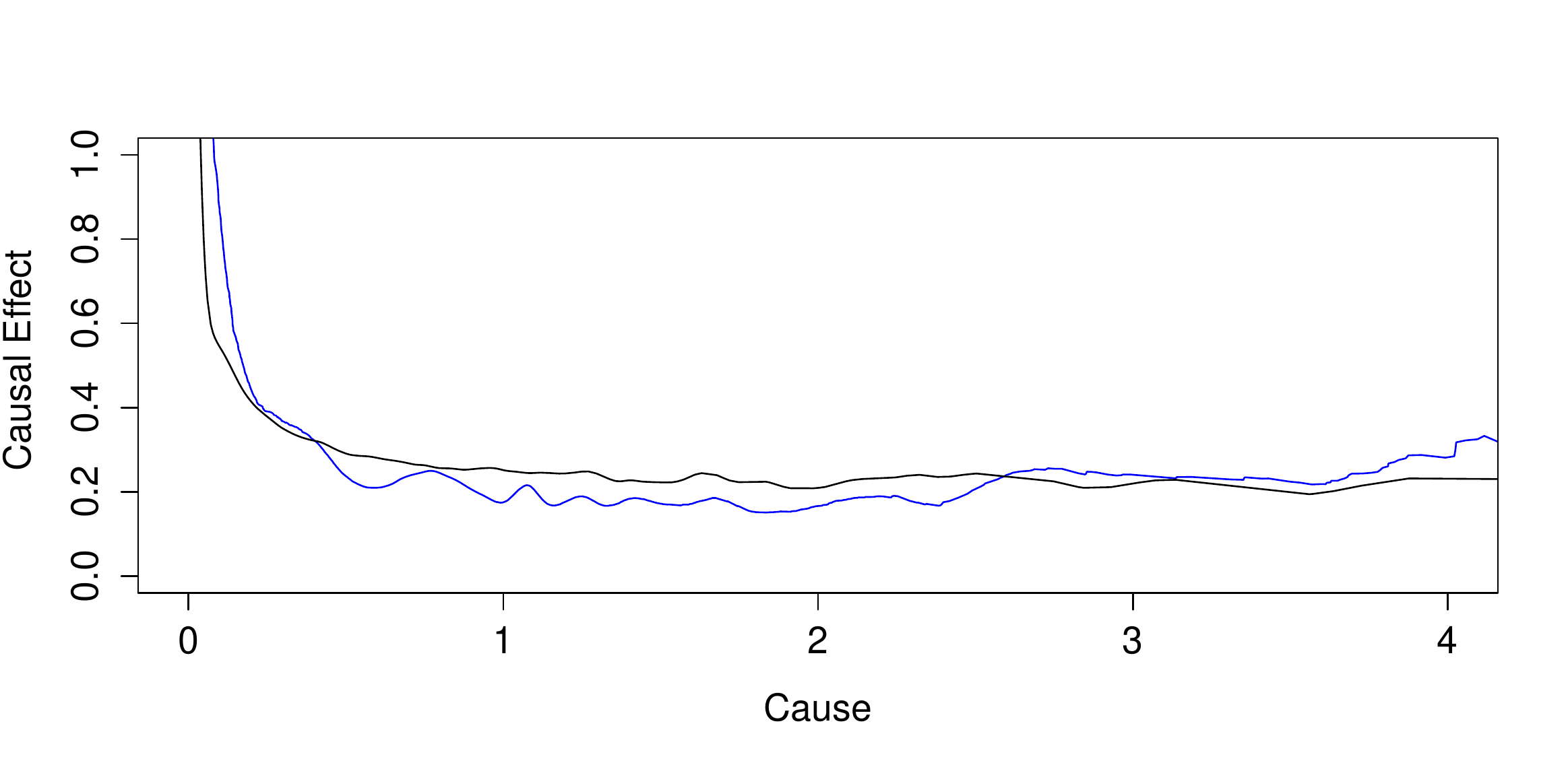}
     \label{EXP}
  %}
\caption{Exponential nonparanormal simulation: black line shows the true causal effect and the blue line represents the causal effect estimated by our NCE method. } % and 20 random
  %variable orderings per graph
  \label{EXP}
\end{figure}

%%%%%%%%%%%%%%%%%%%%%%%%%%%%%%%%%%%%%%%%%%%%%%%%%%%%%%%%%%%%%%%%%%%%%%%%%%%%%%%
\section{TiMet: circadian regulation in \emph{Arabidopsis Thaliana}}
\label{sec:real data}
%%%%%%%%%%%%%%%%%%%%%%%%%%%%%%%%%%%%%%%%%%%%%%%%%%%%%%%%%%%%%%%%%%%%%%%%%%%%%%%
In this section, we illustrate our proposed approach by applying it to a time course gene expression dataset related to the
study of circadian regulation in plants. The data used in our study come from the EU project TiMet (FP7 245143, 2014), whose objective is the elucidation
of the interaction between circadian regulation and metabolism in plants.

\begin{figure}[!h]\centering%
   % \subfigure[]{
     \includegraphics[scale=.3,angle=0.5]{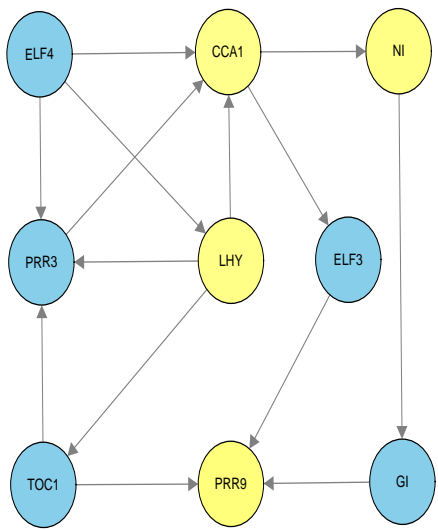}
       %}
\caption{The inferred causal network among the circadian clock genes for  \emph{Arabidopsis thaliana}. Yellow nodes refer to morning genes and blue nodes to evening genes. } % and 20 random
  %variable orderings per graph
  \label{fig:arab DAG}
\end{figure}
%%%%%%%%%%%%%%%%%%%%%%%%%%%%%%%%%%%%%%%%%%%%%%%%%%%%%%%%%%%%%%%%%%%%%%%%%%%%%%%
%\subsection{Application to Arabidopsis thaliana dataset}
%\label{sec: Arabidopsis thaliana dataset}
%%%%%%%%%%%%%%%%%%%%%%%%%%%%%%%%%%%%%%%%%%%%%%%%%%%%%%%%%%%%%%%%%%%%%%%%%%%%%%%

The data consist of transcription
profiles for the core clock genes from the leaves of various genetic variants of \emph{Arabidopsis Thaliana}, measured with qRTPCR. The transcription profiles of the core clock genes  \citep{aderhold2014statistical,pokhilko2010data,guerriero2012stochastic} were recorded: LHY, CCA1, PRR3, NI (PRR5), PRR9, TOC1, ELF3, ELF4 and GI. The plants were grown in
the following 3 light conditions: a diurnal cycle with 12 hr light and 12 hr darkness (12L/12D), an extended night
with full darkness for 24 hrs, and an extended light with constant light for 24 hrs. An exception is the
ELF3 mutant, which was grown only in 12L/12D condition. Samples were taken every 2 hrs to measure mRNA concentrations.
We consider the same group of nine genes, which from previous studies are known to be
involved in circadian regulation \citep{grzegorczyk2011improvements,grzegorczyk2011non, grzegorczyk2008modelling, jia2009analysis}. They consist of two groups of genes: “Morning genes”, which are
LHY, CCA1, PRR9, and PRR5, whose expression peaks in the morning, and  “Evening genes”, including TOC1,
ELF4, ELF3, GI, and PRR3, whose expression peaks in the evening. The expressions for all the genes are strictly positive and highly right-skewed. 

In traditional analysis of microarray data, the data are typically log-transformed. Especially when using the data for prediction, such transformations are sensible as they typically stabalize variances and make down-stream analyses more robust. In our case, however, our aim is to describe the system. We are \emph{not} interested in the causal effect of the log-transformed variables, but we are interested in the causal effects of the original variables. For this reason, we consider the raw data directly, since this is the scale on which we would like to evaluate the system. 

For inferring the underlying causal CPDAG, we considered the RPC-algorithm in the version that uses the
Kendall’s  tau -- results using  Spearman's rho were almost the same. The CPDAG contains three Markov equivalent DAGs.
One of these three causal networks among the genes is displayed in Figure \ref{fig:arab DAG}. For all three causal DAGs, we infer the causal effects between the genes and these are shown as three lines in each of the plots in Figure \ref{fig:causal for arab}. A striking feature is that most of the causal effects shrink towards zero for large values of the cause.  
% In these figures morning genes are represented by lightly shaded circles and evening genes by
%shaded circles. 

%In the  Figure  \ref{fig:arab DAG}, there are several directed genes pointing from morning genes
%to evening genes and vise-versa. Some of the genes play important roles in the circadian clock network.
%In this work, we analysed the causal effect between genes and apply our method with nongaussian assumption. some important genes causal effect are in the Figure \ref{fig:causal for arab}.

The morning gene CCA1 was found to repress the evening genes EFL3 and NI. Among the evening genes, EFL4 and TOC1 have the strongest effect on both other evening and morning genes. The evening gene ELF has positively affects CCA1. It also has a negative effect on LHY. Moreover, the evening genes ELF3, GI and TOC1 are involved in the activation of the  morning
gene PRR. The morning gene LHY has a almost constant effect on the evening genes ELF4, TOC1 and EFL4. In particular ELF4 interacts positively  
with NI and CCA1 and negatively with LHY. Many of these results
are consistent with the findings in \cite{grzegorczyk2011improvements,grzegorczyk2011non}, \cite{aderhold2014statistical} and references therein, as well as with 
the biological network referred to in \cite{jia2009analysis}.
%,  which is based on the work of Mas (2008) and Salome and McClung (2004).

\begin{figure}[!h]\centering%
   % \subfigure[]{
     \includegraphics[scale=.7,angle=0]{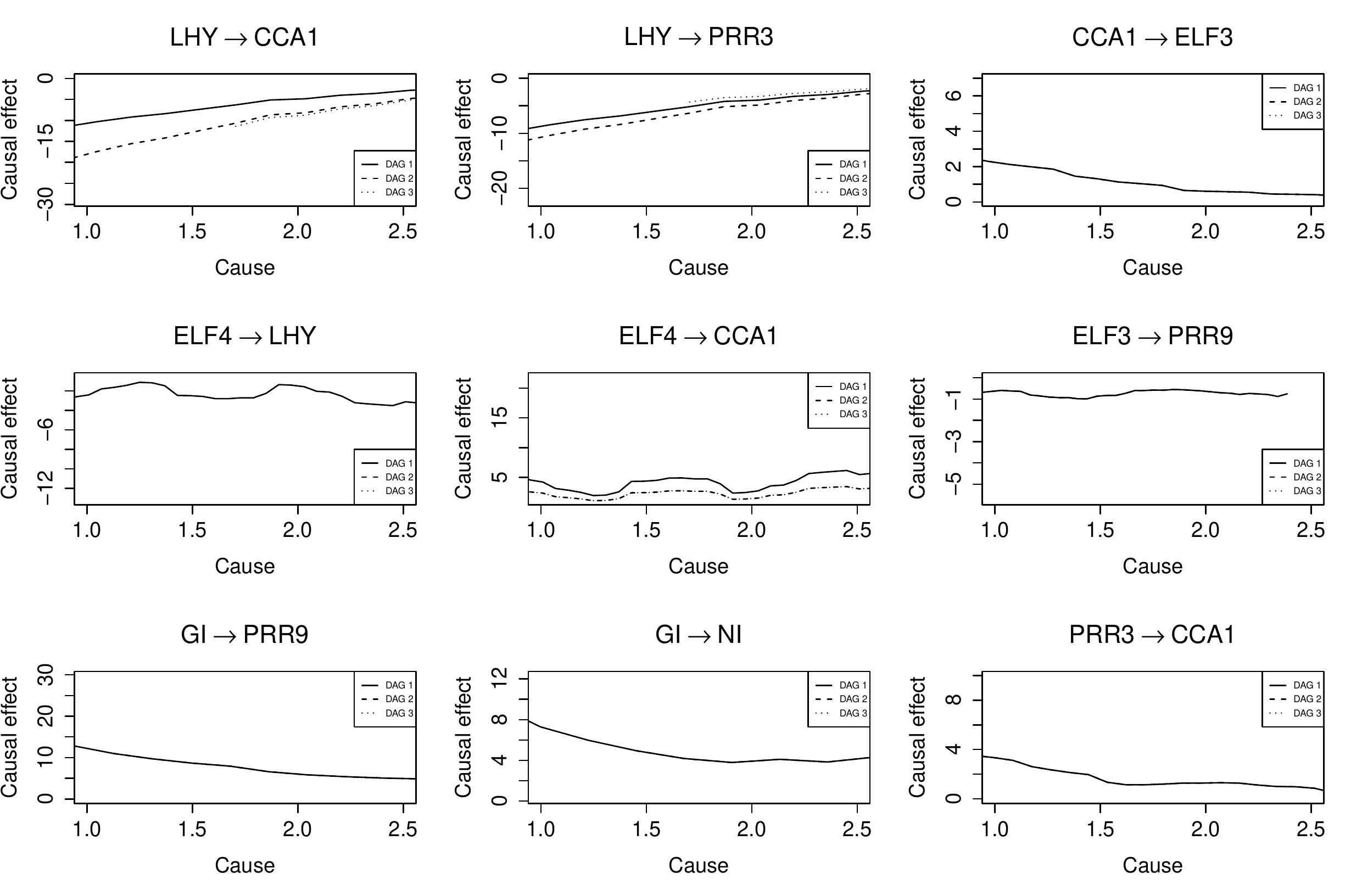}
       %}
\caption{Causal effects for the circadial gene interaction network in \emph{Arabidopsis thaliana}. Whereas ELF3 and ELF4 have almost constant causal effects, the others have a distinctive shrinkage in their causal effects for larger values of the cause. } % and 20 random
  %variable orderings per graph
  \label{fig:causal for arab}
\end{figure}

%\bibliography{MahErnst14.bbl}

%%%%%%%%%%%%%%%%%%%%%%%%%%%%%%%%%%%%%%%%%

\section{Conclusion}
\label{s:discuss}

In this paper, we have derived an explicit formula for describing a causal effect for a flexible class of distributions, the nonparanormal. These distributions are especially useful for real-life observational studies, where normality assumptions are often not warranted. We  presented a simple method, NCE, to estimate these causal effects nonparametrically, based on a first order approximation of the general causal effect formula. It is able to capture a large range of non-linear causal effect. In a simulation study, we have shown that the estimation method works well, particularly away from the tails of the data. We have also applied the method to an \emph{Arabidopsis Thaliana} circadian clock network. The estimated causal effects all reveal a tendency for the causal effects to shrink to zero for large values of the cause, which means that gene regulation shows effect saturation for high levels of the regulator. This is in correspondence with simple Michaelis-Menten kinetic models, often used to model gene regulation.

\appendix
\section{Appendix: calculation of equation (\ref{eqn:Gaussian copula for exp})}

%\section{Appendix section}\label{app}

%\section{}
%\subsection{Calculation of Equation (\ref{eqn:Gaussian copula for exp})}
If we have a bivariate Gaussian copula, with dependence parameter $\rho$, we have
\begin{eqnarray*}\nonumber
 E(Y |X = x) &=& \int_{\mathbb{R}} y {\frac{\partial}{\partial y}} \Phi(\frac{\Phi^{-1}(F(y))-\rho \Phi^{-1}(G(x))}{\sqrt{1-\rho^ 2}})dy \\
\end{eqnarray*}
We choose both marginal distributions $F(y)$ and $G(x)$  are exponential with parameter $\lambda_y,\lambda_x >0$, respectively. Hence,
 \begin{eqnarray*}
 E(Y |X = x)&=& \int_{\mathbb{R}} y  \phi(\frac{\Phi^{-1}(F(y))-\rho \Phi^{-1}(G(x))}{\sqrt{1-\rho^ 2}})\frac{1}{{\sqrt{1-\rho^ 2}}}{\frac{\partial}{\partial y}}\Phi^{-1}(F(y))dy\nonumber\\
  &=&\frac{1}{{\sqrt{1-\rho^ 2}}} \int_{\mathbb{R}} y  \phi(\frac{\Phi^{-1}(F(y))-\rho \Phi^{-1}(G(x))}{\sqrt{1-\rho^ 2}})\frac{1}{\phi(\Phi^{-1}(F(y)))}{\frac{\partial}{\partial y}}(F(y))dy\nonumber\\
&=&\frac{1}{{\sqrt{1-\rho^ 2}}} \int_{\mathbb{R}} y \phi(\frac{\Phi^{-1}(F(y))-\rho \Phi^{-1}(G(x))}{\sqrt{1-\rho^ 2}})\frac{f(y)}{\phi(\Phi^{-1}(F(y)))}dy\nonumber\\
 &=&\frac{1}{\sqrt{1-\rho ^2}}\int_{\mathbb{R}} y \phi(\frac{\Phi^{-1}(1-\exp(\lambda_y y))-\rho \Phi^{-1}(1-\exp(\lambda_x x))}{\sqrt{1-\rho^ 2}}) \frac{{\exp({-\lambda_y y})}}{\phi(\Phi^{-1}(1-\exp(\lambda_y y)))} dy \nonumber 
 \end{eqnarray*}
Thefore, for a bivariate nonparanormal with exponential marginals, we find the following causal effect. Let assume $t=\frac{\Phi^{-1}(1-\exp(\lambda_y y))-\rho \Phi^{-1}(1-\exp(\lambda_x x))}{\sqrt{1-\rho^ 2}}$,
 \begin{eqnarray*}
\mbox{CE}(Y| X_i = x)&=&{\frac{\partial}{\partial x}}E(Y |X = x)\nonumber\\
&=&\frac{1}{\sqrt{1-\rho ^2}}\int_{\mathbb{R}}y{\frac{\partial}{\partial x}}\phi(t) \frac{{\exp({-\lambda_y y})}}{\phi(\Phi^{-1}(1-\exp(\lambda_y y)))} dy \nonumber\\
&=&\frac{1}{\sqrt{1-\rho ^2}}\int_{\mathbb{R}} y \phi^{'}(t)\frac{-\rho\ {\exp({-\lambda_x x})}}{\sqrt{1-\rho ^2}\phi(\Phi^{-1}(1-\exp(\lambda_x x)))} \frac{{\exp({-\lambda_y y})}}{\phi(\Phi^{-1}(1-\exp(\lambda_y y)))} dy \nonumber\\
&=& - \frac{\rho}{1-\rho ^2}\int_{\mathbb{R}} y\phi^{'}(t)\frac{{\exp({-\lambda_x x})}{\exp({-\lambda_y y})}}{{\phi(\Phi^{-1}(1-\exp(\lambda_y y)))}{\phi(\Phi^{-1}(1-\exp(\lambda_x x)))}} dy\nonumber
\end{eqnarray*}
where $\Phi$ and $\phi$ are cumulative distribution and  density function of Gaussian, respectively.
\label{lastpage}

\bibliography{causal-effect}

\end{document}